\documentclass[lettersize,journal]{IEEEtran}
\IEEEoverridecommandlockouts
\usepackage{amsmath,amsfonts,amsthm,amssymb}
\usepackage{array}
\usepackage{textcomp}
\usepackage{subfigure}
\usepackage{float}
\usepackage{enumerate}
\usepackage{tabularx}
\usepackage{stfloats}
\usepackage{url}
\usepackage{hyperref}
\usepackage{xcolor,soul,framed} 
\usepackage{pstricks}
\usepackage{verbatim}
\usepackage{graphicx}
\usepackage{epstopdf}
\usepackage{booktabs}
\usepackage{cite}
\usepackage[T1]{fontenc}
\usepackage{tikz}
\usepackage{algorithm}
\usepackage{algpseudocode}
\usepackage{setspace}
\usepackage{booktabs}
\usepackage{makecell}
\usepackage{multirow}
\allowdisplaybreaks
\begin{document}
\newtheorem{Remark}{Remark}
\newtheorem{remark}{Remark}
\newtheorem{thm}{Theorem}
\newtheorem{lem}{Lemma}
\renewcommand{\algorithmicrequire}{\textbf{Input:}} 
\renewcommand{\algorithmicensure}{\textbf{Output:}}

\title{
Flexible Multi-Beam Synthesis and Directional Suppression Through Transmissive RIS
\thanks{Manuscript received xxxxxx; revised xxx, 2024; accepted xxx, 2024.}
\thanks{This work is supported by the National Natural Science Foundation of China (NSFC) under Grant NO.12141107, and the Interdisciplinary Research Program of HUST (2023JCYJ012) ({\it Corresponding author: Ke Yin}).}
}

\author{Rujing~Xiong,~\IEEEmembership{Graduate Student Member,~IEEE,}    
        Ke Yin,  
        Jialong~Lu, 
        Kai Wan,~\IEEEmembership{Member,~IEEE,} 
        Tiebin Mi,~\IEEEmembership{Member,~IEEE,} and
        Robert~Caiming~Qiu,~\IEEEmembership{Fellow,~IEEE}
\thanks{R.~Xiong, J. Lu, K. Wan, T. Mi, and R. Qiu are with the School of Electronic Information and Communications, Huazhong University of Science and Technology, Wuhan 430074, China (e-mail: \{rujing@hust.edu.cn, rujingxiong@gmail.com\}, \{jialong, mitiebin, caiming\}@hust.edu.cn).}
\thanks{K. Yin is with the Center for Mathematical Sciences, Huazhong University of Science and Technology, Wuhan 430074, China (e-mail: kyin@hust.edu.cn).}

}
\maketitle

\maketitle

\begin{abstract}
Despite extensive research on reconfigurable intelligent surfaces (RISs) in recent years, existing beamforming methods still face significant challenges in achieving flexible and robust beam synthesis, which is an essential capability for a wide range of communication scenarios. This paper introduces a Max-min criterion with nonlinear constraints, leveraging optimization techniques to simultaneously enable flexible multi-beam synthesis and directional suppression using transmissive RIS. Firstly, a realistic model grounded in geometrical optics is introduced to characterize the input/output behaviors of transmissive RISs, effectively bridging the gap between explicit beamforming requirements and practical implementations. Subsequently, a highly efficient algorithm for constrained Max-min optimizations involving quadratic forms is developed. By introducing an auxiliary variable and applying the compensated convexity transform, we successfully reformulate the original non-convex problem and obtain the optimal solution iteratively. This approach is readily applicable to a wide range of constrained Max-min optimization problems. Finally, numerical simulations and prototype experiments are conducted to validate the effectiveness of the proposed framework. The results demonstrate that the proposed algorithm can effectively enhance or selectively suppress signal beams in designated spatial directions, outperforming existing methods in terms of beam control accuracy and robustness. This framework provides valuable insights and references for practical communications applications such as physical layer security and interference mitigation.

\end{abstract}

\begin{IEEEkeywords}
Reconfigurable intelligent surface (RIS), beam synthesis, directional suppression, Max-min, non-convex optimization, prototype experiment.
\end{IEEEkeywords}

\section{Introduction}
\IEEEPARstart{R}{econfigurable} intelligent surface (RIS) is anticipated to be a pivotal technology for future 6G communication networks due to its great capability to reconfigure and optimize wireless propagation environments~\cite{cui2014coding,basar2019wireless,wu2023intelligent}. A RIS consists of a number of well-designed passive units, each capable of independently manipulating the characteristics of incident electromagnetic (EM) waves, such as phase and amplitude. By dynamically adjusting these factors in response to environmental conditions, RIS enables the redistribution of incident waves to formulate directional beams, 
thereby significantly enhancing coverage, transmission rates, and overall network performances~\cite{xiong2024fair}. Additionally, since RIS has no dedicated signal processing module, it is less susceptible to interference and noise, which is beneficial for signal transmission. Moreover, compared to a relay, the passive nature of the RIS makes it simpler in structure and requires fewer complex components, leading to lower system power consumption and cost. Due to its distinctive features, RIS has been widely employed across various wireless systems for enhanced performance.

In wireless communications, to tackle the limitations in scenarios such as outdoor/indoor through walls or windows and outside/inside a vehicle~\cite{tang2023transmissive,li2022coverage,kitayama2021transparent}, the transmissive RIS technique is garnering increasing attention~\cite{docomo2020docomo,bai2020high,wan2022space,zhang2024design}. Unlike the widely studied reflective RIS, the transmissive RIS allows signals to transmit through itself to form directional beams. Meanwhile, similar to reflective RISs, transmissive RISs can be utilized to establish extra refractive communication links between the base station (BS) and user equipment (UE). Furthermore, owing to the absence of physical topology constraints~\cite{wang2023doppler}, transmissive RIS is increasingly utilized to design new transmitter architectures, facilitating beamforming design and power allocation in both cell~\cite{grossi2024beampattern,li2021beamforming,sun2022ris,jamali2020intelligent,li2023robust,li2023towards} and cell-free networks~\cite{demir2024usercentric}. In such architectures, transmissive RIS is deployed in the front of the BS’s feeds to serve as both the analog-precoder networks and transmit antennas, such that the system energy efficiency~\cite{sun2022energy} and spectral efficiency~\cite{sun2022ris} can be enhanced with low hardware cost. 
Additionally, the transmissive RIS boasts high aperture efficiency due to its structural features, which eliminates the issues of feed source occlusion and the interference from surface-reflected waves~\cite{li2024transmissive,li2021beamforming, liu2023transmissive}. Moreover, transmissive RIS can be stacked to form a structure similar to artificial neural networks, enabling advanced signal processing directly in the native EM waves domain~\cite{wang2024multiuser,niu2024stacked}. These factors render the transmissive RIS efficient in operation and position it with great promise for future communication networks~\cite{wan2022space}. 

Beam synthesis characterizes the process through which RISs convert the energy from incident electromagnetic waves into reflected or refracted waves, explicitly forming directional beams~\cite{xiong2024fair}. As a primary function of a reconfigurable passive antenna array, it determines the pattern of the radiation waves and significantly impacts the signal-to-noise ratio (SNR) and capacity of the communication systems, thereby playing a fundamental and critical role in signal coverage, target detection, location, and tracking, as well as in anti-eavesdropping and anti-interference communications~\cite{wu2023intelligent,xiong2024fair}. In both reflective RIS and transmissive RIS, the simple beam synthesis function can be achieved through maximal ratio combining (MRC), such as phase configuration~\cite{han2020half,tang2023transmissive}.

Additionally, given the similarity of mathematical models, the phase configuration methods utilized in reflective RIS can be readily applied to transmissive RIS systems~\cite{tang2023transmissive}. Nevertheless, existing methods are limited to achieving only simple beam synthesis functionality. In~\cite{tang2023transmissive,han2020cooperative,ma2023multi}, single beam aggregation and beam scanning were achieved through phase compensation methods~\cite{stutzman2012antenna}, where the units of reflective RIS or transmissive RIS are configured to compensate for the phase differences introduced by different path lengths from the transmitter and receiver to each unit on RIS. Additionally, exhaustive search~\cite{wang2022beam}, hierarchical search~\cite{xiao2016hierarchical,wang2023hierarchical}, and machine learning~\cite{heng2021machine} methods were also employed to deal with the simple beam synthesis like beam alignment.

For achieving complex beam synthesis, the authors in~\cite{droulias2024reconfigurable} presented a design of the RIS as a spatial filter based on transfer function principles by modifying its shape and size. While this method facilitated various beam syntheses through the RIS, it altered the hardware foundation, leading to limited flexibility. Additionally,~\cite{rahal2022arbitrary} and~\cite{rahal2023performance} approximated the desired beam pattern by the method based on a predefined lookup table. In~\cite{meng2023efficient}, the author used the superposition of the reflection coefficient associated with each beam for multi-beam generation. The authors in~\cite{grossi2024beampattern} investigated the least squares amplitude beampattern matching method for transmitter beampattern design, indicating that for the same desired radiation pattern, the RIS can achieve comparable performance with a full-digital multiple-input multiple-output (MIMO) array. Furthermore, in our previous work~\cite{xiong2024fair}, we have demonstrated that through proper phase configuration, the RIS can implement flexible beam synthesis such as beam-splitting and wide-beam generation, and control the received signal power at different users. 

However, the aforementioned works primarily focused on the direction of the main beams, without considering the impact of sidelobe energy. Sidelobes may introduce additional interference and noise, potentially degrading system performance in complex networks. Furthermore, elevated sidelobe levels heighten the risk of unintended signal leakage, posing potential threats to information security. Consequently, effective sidelobe suppression and flexible beam control are crucial for enhancing interference resilience and physical layer security in wireless communication systems.

\subsection{Contributions}

In this work, we address the flexible beam synthesis challenge by incorporating sidelobe constraints into the RIS beamforming design. Through formulating and solving constrained Max-min problems, the proposed framework allows the target signal to be enhanced or suppressed in specified spatial directions, leading to explicitly realized~\textbf{high-directional beams} or \textbf{beam nulls}. 
To our knowledge, this represents the first successful implementation of simultaneous flexible beam enhancement and directional suppression using a passive antenna array solely through phase configuration. 

The main contributions can be summarized as follows:

\begin{itemize}
\item \textbf{Physical model and constrained Max-min problem formulation}. 
Based on the electromagnetic response of the RIS units~\cite{xiong2024fair,mi2023towards}, we present practical and realistic models grounded in geometrical optics to characterize the multiple input/output behaviors of transmissive RISs in the radiation field. The MIMO signal model captures the phase and loss variations resulting from differences in the path lengths of the incident and scattered signals across different units. Furthermore, considering the existence of unauthorized users in the system, a constrained Max-min optimization problem is formulated for exploring the flexible beam synthesis capabilities, including directional suppression, at the transmissive RIS side.

\item \textbf{New method for the optimal solution}. To efficiently solve the constrained Max-min problem, we introduce an auxiliary variable, along with a smooth approximation based on the compensated convexity transform, to reformulate the original non-convex optimization problem. The approximation function is tight and continuous differentiable with Lipschitz continuous gradient (i.e., $C^{1,1}$ smooth). A bisection-based (BIS) algorithm is then proposed, which iteratively converges to the optimal solution with a computational complexity of $\mathcal{O}(\log L)$, where $L$ denotes the size of the search space.
In each iteration, a smooth unconstrained sub-problem is solved using the accelerated gradient descent method. The proposed approach is related to the Majorization Minimization (MM) algorithmic framework. It is independent of specific signal models and exhibits excellent extensibility, making it readily applicable to a broad class of constrained Max-min optimization problems.

\item \textbf{Flexible beam synthesis exploration and achievement}. 
We investigate the beam redistribution capabilities of transmissive RISs using the proposed framework. Within this framework, energy distribution across different directional beams can be flexibly adjusted, allowing for precise signal enhancement or physical suppression.  
Moreover, the proposed algorithm is readily applicable to reflective and other types of passive RISs by adopting appropriate modeling strategies and optimization problem formulations.
The effectiveness of the flexible beam synthesis framework is validated through extensive simulations and prototype measurements. 
\end{itemize}

\subsection{Outline}
The remainder of the paper is organized as follows. In Section~\ref{Section2}, we present the modeling of transmissive RIS-aided MIMO communications and formulate the beam synthesis problems for maximizing the minimum received signal power among target communication users, with the constraint at certain directions related to unauthorized users like eavesdroppers. Section~\ref{Section3} proposes an efficient algorithm based on compensated convexity transform and bisection method to solve the problem. In Section~\ref{Section4}, we validate the efficacy of the proposed framework and evaluate the performance in beam synthesis and received signal power through comprehensive simulations and practical measurement experiments. Finally, the paper is concluded in Section~\ref{Section5}.

\subsection{Reproducible Research}
The results are reproducible with codes at: \url{https://github.com/RujingXiong/RIS\_BeamSynthesisAndSuppression}

\subsection{Notations}
The imaginary unit is denoted by $j$. The magnitude component of a complex number is represented by $|\cdot|$. Unless explicitly specified, lower and upper case bold letters denote vectors and matrices. The conjugate transpose, conjugate, and transpose of~$\mathbf{A}$ are written as $\mathbf{A}^{\mathrm H}$, $\mathbf{A}^{*}$ and $\mathbf{A}^{\mathrm T}$, respectively. 

\section{System Model and problem formulation}\label{Section2}
Due to the unique structural characteristics, transmissive RIS is frequently utilized in transmitter architectures, as illustrated in Fig.~\ref{F3}. In this section, we consider a transmissive RIS-aided multi-user downlink system with the presence of unauthorized users, such as eavesdroppers or interfering users from adjacent cells. This consideration captures scenarios where a transmissive RIS functions as a part of the transmitter or a passive relay to establish additional refractive links, with the key distinction being the varying distances from the sources of incident signals to the RIS~\footnote{For analysis, we focus on radiated waves and do not consider waves with turbulent electric field effects in the reactive near-field region~\cite{yaghjian1986overview,liu2023near}.}.
\subsection{System Model}

\begin{figure}
  \centering
  \includegraphics[width=0.9\linewidth]{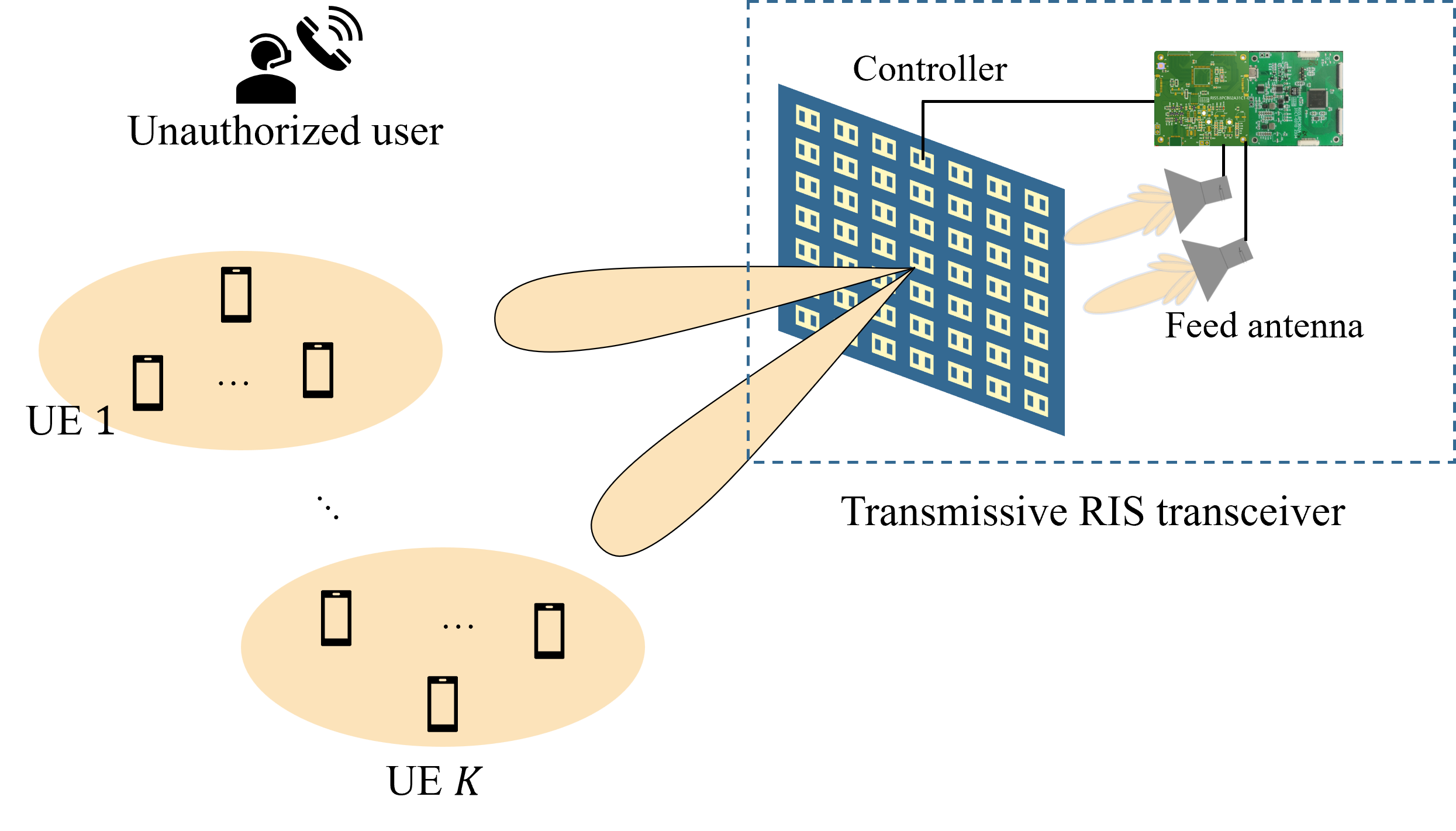}
  \caption{The communication system involving a transmitter architecture using transmissive RIS}
  \label{F3}
\end{figure}

As we concentrate on an arbitrary transmissive RIS situated on the $xoy$-plane, and its geometric center is located at the origin of the coordinate system, as shown in Fig.~\ref{F2}. The transmissive RIS is a uniformly planar array consisting of $N$ units located at $\mathbf{p}_n = [x_n,y_n,z_n]^\mathrm{T}, \ n = 1,\dots,N$. The unit spacing along both the $x$-axis and $y$-axis is $d$, usually set to half of the wavelength. Assume the RIS is illuminated by a single incident EM wave with electric field $E(r^\text{t}, \theta^\text{t}, \phi^\text{t})$, originating from the source point ($r^\text{t}, \theta^\text{t}, \phi^\text{t}$). Here symbols $r^\text{t} \in \mathbb{R}$, $\theta^\text{t} \in [0^{\circ},90^{\circ})$, and $\phi^\text{t} \in [0^{\circ},360^{\circ})$ denote the distance, the elevation angle, and the azimuth angle. As the EM wave propagates towards the $n$-th unit, the attenuation behavior is characterized by the factor $e^{ - j 2 \pi  r^\text{t}(n) / \lambda}/ r^\text{t}(n)$. Considering the unit on the transmissive RIS is isotropic, with scattering pattern $g$~\cite{mi2023towards}, the scattered field at the unit can be expressed as $E e^{ - j 2 \pi  r^\text{t}(n) / \lambda}/ r^\text{t}(n)g e^{\Omega_{n}}$. Similarly, involving propagation towards the communication UE, the electric field at the observation point along the $n$-th unit can be expressed as
\begin{multline}
E(n)(r^\text{r}, \theta^\text{r}, \phi^\text{r}) \\= g e^{\Omega_{n}} \frac{e^{ - j 2 \pi  r^\text{t}(n) / \lambda}e^{ - j 2 \pi r^\text{r}(n) / \lambda}}{r^\text{t}(n)r^\text{r}(n)}E(r^\text{t}, \theta^\text{t}, \phi^\text{t}),    
\end{multline} 
where $\lambda$ denotes the wavelength, $ r^\text{t}(n)$ and $r^\text{r}(n)$ represents the distance from the $n$-th unit of the RIS to the transmitter (Tx) and the receiver (Rx), and $\Omega_{n}$ represents the phase configuration of the $n$-th unit~\footnote{ Here we assume that the transmissive RIS units induce only phase shifts to the incident EM waves, while maintaining their amplitude unchanged.}. Thus, the aggregated electric field at the observation point is determined by the superposition of individual fields scattered by $N$ units, given by
\begin{multline}
E(r^\text{r}, \theta^\text{r}, \phi^\text{r}) \\=
\sum^N_{n =1} g e^{\Omega_{n}} \frac{e^{ - j 2 \pi  r^\text{t}(n) / \lambda}e^{ - j 2 \pi r^\text{r}(n) / \lambda}}{r^\text{t}(n)r^\text{r}(n)}E(r^\text{t}, \theta^\text{t}, \phi^\text{t}).
\end{multline}
When the illumination signal sources originate from multiple directions with strengths denoted by $ \{ E (r^\text{t}_m, \theta^\text{t}_m,\phi^\text{t}_m), m=1, \ldots, M\}$, the observed electric field is the superposition of multiple electric fields, as
\begin{equation}\label{Model_MISO}
\begin{aligned}
  E (r^\text{r}, \theta^\text{r}, \phi^\text{r})= 
  & \sum_{m=1}^{M} E ( r^\text{t}_m, \theta^\text{t}_m, \phi^\text{t}_m )  \\
  & \sum_{n=1}^{N} g  e^{j \Omega_n} \frac{ e^{ - j 2 \pi r^\text{t}_m (n) / \lambda} e^{ - j 2 \pi r^\text{r}(n) / \lambda } }{ r^\text{t}_m (n)  r^\text{r}(n) }   . 
\end{aligned}
\end{equation}
\begin{figure}[tbp]
  \centering
  \includegraphics[width=1\linewidth]{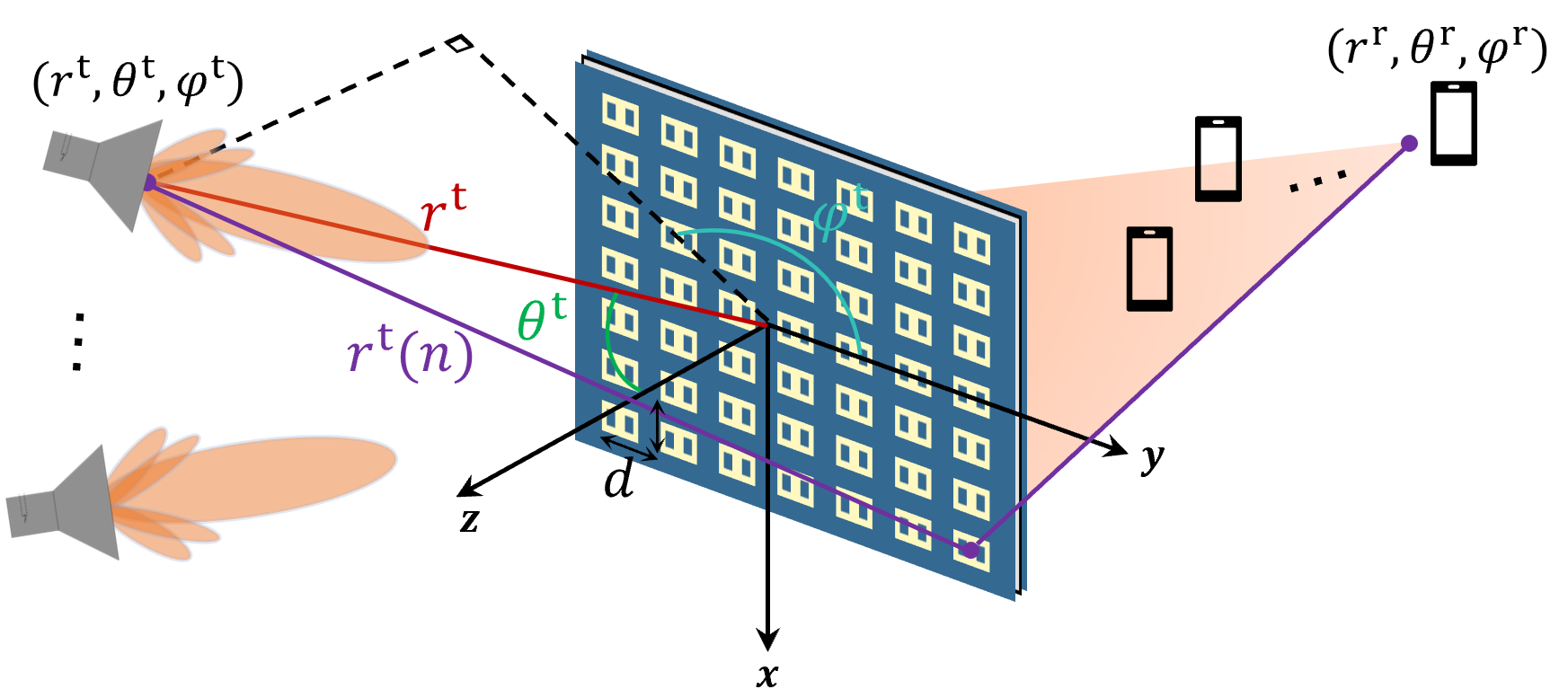}
  \caption{System model of transmissive RIS-assisted communications.}
  \label{F2}
\end{figure}
Suppose the amplitude attenuation from the single source to each RIS unit is virtually identical, and there are $K$ UEs in the far field. The observed electric field at different UE observation points can be expressed as in~\eqref{E:StructuredModel}.
\begin{figure*}[tbp]
  \begin{equation}\label{E:StructuredModel}
    \begin{aligned}
\underbrace{
        \begin{bmatrix}
        E (r^\text{r}_1, \theta^\text{r}_1, \phi^\text{r}_1) \\
        \vdots                           \\
        E (r^\text{r}_K, \theta^\text{r}_K, \phi^\text{r}_K)
      \end{bmatrix} 
 }_{\mathbf{E} (\mathbf{r}^\text{r}, \mathbf{\Theta}^\text{r}, \mathbf{\Phi}^\text{r})}
      = &  g
      \overbrace{
       \begin{bmatrix}
        \frac{ e^{-j 2 \pi r_1^\text{r} / \lambda }}{ r_1^\text{r} } &        &  0 \\
                               & \ddots &    \\
            0                  &        &  \frac{ e^{-j 2 \pi r_K^\text{r} / \lambda }}{ r_K^\text{r} }
       \end{bmatrix}}^{\mathbf{L} (\mathbf{r}^\text{r})}
 \overbrace{
      \begin{bmatrix}
        e^{ j 2 \pi \mathbf{p}_1^{\mathrm{T}} \mathbf{u} (\theta_{1}^\text{r}, \phi_{1}^\text{r})  / \lambda } & \cdots & e^{ j 2 \pi \mathbf{p}_N^{\mathrm{T}} \mathbf{u} (\theta_{1}^\text{r}, \phi_{1}^\text{r})  / \lambda } \\ 
\vdots & \ddots & \vdots\\
e^{ j 2 \pi \mathbf{p}_1^{\mathrm{T}} \mathbf{u} (\theta_{K}^\text{r}, \phi_{K}^\text{r})  / \lambda } & \cdots & e^{ j 2 \pi \mathbf{p}_N^{\mathrm{T}} \mathbf{u} (\theta_{K}^\text{r}, \phi_{K}^\text{r})  / \lambda }
      \end{bmatrix} }^{\mathbf{V} (\mathbf{\Theta}^\text{r}, \mathbf{\Phi}^\text{r}) }
      \overbrace{\begin{bmatrix}
        e^{j \Omega_1} &        &  0 \\
                               & \ddots &    \\
            0                  &        & e^{j \Omega_N}
      \end{bmatrix}}^{\text{diag} (e^{j \mathbf{\Omega} })} \\
      & \underbrace{\begin{bmatrix}
        e^{ j 2 \pi \mathbf{p}_1^{\mathrm{T}} \mathbf{u} (\theta_{1}^{\text{t}}(1), \phi_{1}^\text{t}(1))  / \lambda }  & \cdots  & e^{ j 2 \pi \mathbf{p}_1^{\mathrm{T}} \mathbf{u} (\theta_{M}^\text{t}(1), \phi_{M}^\text{t}(1))  / \lambda } \\
        \vdots  & \ddots & \vdots \\
        e^{ j 2 \pi \mathbf{p}_N^{\mathrm{T}} \mathbf{u} (\theta_{1}^{\text{t}}(N), \phi_{1}^\text{t}(N))  / \lambda }  & \cdots  & e^{ j 2 \pi \mathbf{p}_N^{\mathrm{T}} \mathbf{u} (\theta_{M}^\text{t}(N), \phi_{M}^\text{t}(N))  / \lambda }\\
      \end{bmatrix}}_{\mathbf{V}(\mathbf{\Theta}^\text{t}, \mathbf{\Phi}^\text{t})}
\underbrace{
      \begin{bmatrix}
        \frac{ e^{-j 2 \pi r^\text{t}_{1} / \lambda } }{ r^\text{t}_{1} } &        &  0 \\
                               & \ddots &    \\
            0                  &        &  \frac{ e^{-j 2 \pi r^\text{t}_{M} / \lambda } }{ r^\text{t}_{M}}
      \end{bmatrix}}_{\mathbf{L} (\mathbf{r}^\text{t}) }
     \underbrace{ \begin{bmatrix}
        E (r_1^\text{t}, \theta_{1}^\text{t}, \phi_{1}^\text{t}) \\
        \vdots    \\
        E (r_M^\text{t}, \theta_{M}^\text{t}, \phi_{M}^\text{t})
      \end{bmatrix}}_{\mathbf{E} ( \mathbf{r}^\text{t}, \mathbf{\Theta}^\text{t}, \mathbf{\Phi}^\text{t} ) }.
    \end{aligned}
  \end{equation}
  \medskip
  \hrule
\end{figure*}
Here, $E(r^\text{r}_k, \theta^\text{r}_k, \phi^\text{r}_k)$ denotes the field at UE $k$, where $k = 1,\cdots, K$. The terms $\theta_{m}^{\text{t}}(n), \phi_{m}^\text{t}(n)$ denote the elevation angle and azimuth angle of the $m$-th incident EM wave relative to the $n$-th unit of the transmissive RIS. $e^{j \mathbf{\Omega}} = [e^{j \Omega_1}, \cdots, e^{j\Omega_N}]^{\mathrm{T}}$ represent the phase configuration of the RIS. The direction vector is defined as $\mathbf{u} ( \theta, \phi) = [\sin \theta \cos \phi, \ \sin \theta \sin \phi, \ \cos \theta]^{\mathrm{T}}$.
Due to the similarity in their interaction with the incident signal~\cite{tang2023transmissive}, this signal model is equally applicable to reflective RIS when the RIS is positioned away from the transmitter and operates as a passive relay, with the primary difference being the constraints on the incident and transmission (reflection) angles. However, when the RIS is integrated into the transmitter architecture for beamforming, the model becomes inaccurate due to feed source occlusion and interference from reflected waves.

The advantage of the canonical linear representation \eqref{E:StructuredModel} is that the structure inherent in the channel is identified. The incident and scattered steering matrices $\mathbf{V} (\mathbf{\Theta}^\text{t}, \mathbf{\Phi}^\text{t})\in \mathbb{C}^{K\times N}$ and $\mathbf{V} (\mathbf{\Theta}^\text{r}, \mathbf{\Phi}^\text{r})\in \mathbb{C}^{N\times M}$ capture the phase variations resulting from differences in the path lengths of incident and scattered signals across different units. For uniform linear RISs, these matrices exhibit distinct Vandermonde structures~ \cite{mi2023towards}. Additionally, the diagonal matrices $\mathbf{L} (\mathbf{r}^\text{t})\in \mathbb{C}^{M\times M}$ and $\mathbf{L} (\mathbf{r}^\text{r})\in \mathbb{C}^{K\times K}$ account for the radial propagation of incident and reflected traveling waves from the sources to the RIS and from the RIS to the UEs, respectively. 
\begin{figure}[htbp]
  \centering
  \includegraphics[width=.9\linewidth]{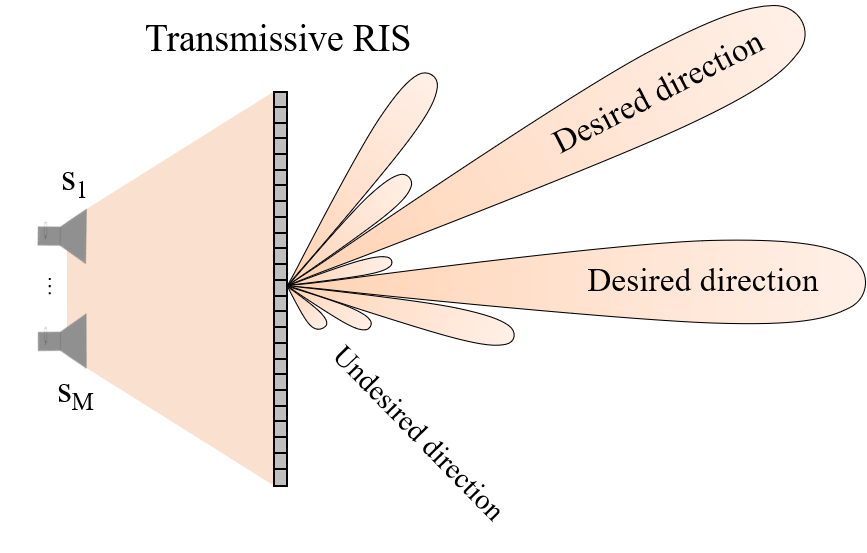}
  \caption{Beam synthesis through Transmissive RIS. The undesired direction may correspond to unauthorized users or users from neighboring cells.}
  \label{F1}
\end{figure}

\subsection{Optimization Formulation}
Given the potential presence of eavesdroppers and/or UEs from neighboring cells (collectively referred to as unauthorized UEs for conciseness in the following), we explore an optimal beam synthesis strategy aimed at maximizing the minimum received power, while imposing constraints to ensure that signal power in specific undesired directions remains below a predefined threshold, as shown in Fig.~\ref{F1}. This approach not only facilitates beam synthesis for multiple users but also provides a solution for addressing information leakage and interference with adjacent cells. Formally, this synthesis strategy can be expressed as 
\begin{equation}\label{FB1}
\begin{aligned}
  \max_{ \Omega_1, \dots, \Omega_{N} \in [0, 2\pi) } \min_{ k } & \  \left \{ P ( r^\text{r}_k, \theta^\text{r}_k, \phi^\text{r}_k ), k = 1, \ldots, K \right \}\\
{\rm s.t.}& \ P ( r^\text{r}_q, \theta^\text{r}_q, \phi^\text{r}_q) \leq \sigma_q, q = K+1,\dots,K+Q.
\end{aligned}
\end{equation}
In this context, we consider $K$ served communication UEs and $Q$ unauthorized UEs, the notation $P (\cdot) = \lvert E(\cdot) \rvert^2$ represents the power at a specific location, and $\sigma_q$ is the threshold value related to the unauthorized UE $q$.

With the input/output models~\eqref{E:StructuredModel}, it can be derived that the received signal electric field at a specific position \((r^\text{r}, \theta^\text{r}, \phi^\text{r})\) can be expressed as  
\begin{equation}
\begin{aligned}  
E &(r^\text{r}, \theta^\text{r}, \phi^\text{r}) \\&= g {{L} ({r}^\text{r})} {\mathbf{V} (\mathbf{\Theta}^\text{r}, \mathbf{\Phi}^\text{r}) } {\text{diag} (e^{j \mathbf{\Omega} })} {\mathbf{V}(\mathbf{\Theta}^\text{t}, \mathbf{\Phi}^\text{t})} {\mathbf{L} (\mathbf{r}^\text{t}) } {\mathbf{E} ( \mathbf{r}^\text{t}, \mathbf{\Theta}^\text{t}, \mathbf{\Phi}^\text{t} ) } \\
& =  \underbrace{
g {{L} ({r}^\text{r})} {\mathbf{V} (\mathbf{\Theta}^\text{r}, \mathbf{\Phi}^\text{r}) } {\text{diag} ({\mathbf{V}(\mathbf{\Theta}^\text{t}, \mathbf{\Phi}^\text{t})} {\mathbf{L} (\mathbf{r}^\text{t}) } {\mathbf{E} ( \mathbf{r}^\text{t}, \mathbf{\Theta}^\text{t}, \mathbf{\Phi}^\text{t} ) } )}
}_{\mathbf{h} (\mathbf{r}^\text{r}, \mathbf{\Theta}^\text{r}, \mathbf{\Phi}^\text{r} ; \mathbf{r}^\text{t}, \mathbf{\Theta}^\text{t}, \mathbf{\Phi}^\text{t} ; \mathbf{E}^\text{t} )^{\mathrm T}\ \in\ \mathbb{C}^N}
e^{j \mathbf{\Omega} }
\end{aligned}
\end{equation}  
Here, $ {L} ({r}^\text{r}) =\frac{ e^{-j 2 \pi r^\text{r} / \lambda }}{ r^\text{r} }$ is a scalar representing the free-space path loss and phase shift from the RIS center to the receiver, while \( {\mathbf{V} (\mathbf{\Theta}^\text{r}, \mathbf{\Phi}^\text{r}) } \) reduces to a $N$-dimensional row vector characterizing the phase differences across different RIS units.

Thus, the received power at the location $(r^\text{r}, \theta^\text{r}, \phi^\text{r})$ can be expressed as
\begin{equation}
 P(r^\text{r}, \theta^\text{r}, \phi^\text{r})= \left \lvert \mathbf{h} (\mathbf{r}^\text{r}, \mathbf{\Theta}^\text{r}, \mathbf{\Phi}^\text{r} ; \mathbf{r}^\text{t}, \mathbf{\Theta}^\text{t}, \mathbf{\Phi}^\text{t} ; \mathbf{E}^\text{t} )^{\mathrm T} e^{j \mathbf{\Omega} } \right \rvert^2
\end{equation}
As vector ${\mathbf{h} (\mathbf{r}^\text{r}, \mathbf{\Theta}^\text{r}, \mathbf{\Phi}^\text{r} ; \mathbf{r}^\text{t}, \mathbf{\Theta}^\text{t}, \mathbf{\Phi}^\text{t} ; \mathbf{E}^\text{t} )}$ is simplified to $\mathbf{h}$ and $\mathbf{w} = [e^{j \Omega_1}, \cdots, e^{j \Omega_N}]^{\mathrm{H}}$ is defined, the received power at the position \((r^\text{r}, \theta^\text{r}, \phi^\text{r})\) can be expressed as
\begin{equation}
P (r^\text{r}, \theta^\text{r}, \phi^\text{r})= \mathbf{w}^{\mathrm{H}} \mathbf{H} \mathbf{w},
\end{equation}
where $\mathbf{H} = \mathbf{h} \mathbf{h}^{\mathrm{H}}$ is a positive semi-definite Hermitian matrix.

Thus, we can rewrite the optimization problem in~\eqref{FB1} as
\begin{equation}\label{II0}
\begin{aligned}
\max_{ \Omega_1, \dots, \Omega_{N} } \min_{1\leq k\leq K} &\ \mathbf{w}^{\mathrm{H}} \mathbf{H}_k \mathbf{w} \\
{\rm s.t.}  &\ \mathbf{w}^{\mathrm{H}} \mathbf{H}_q \mathbf{w} \leq \sigma_q , \  q=K+1,\dots,K+Q. 
\end{aligned}
\end{equation}
Here, $\mathbf{H}_k$ and $\mathbf{H}_q$ denote the signal propagation correlation matrices for the communication UE \(k\) and the unauthorized UE \(q\), respectively.
Furthermore, we can introduce weight factors $\alpha_k $ to control signal power gain among different served UE directions, as we have discussed in our previous work~\cite{xiong2024fair}. For clarity, we define a matrix $\mathbf{A}_k \in \mathbb{C}^{N\times N}$ to incorporate the factor as $\mathbf{A}_k = 1/\alpha_k \mathbf{H}_k$. Finally, the optimization problem to achieve flexible beam synthesis with directional suppression can be formulated as
\begin{equation}\label{MM0}
\begin{aligned}
\text {(P0)} \ &\max_{ \Omega_1, \dots, \Omega_{N} } \min_{1\leq k\leq K} \ \mathbf{w}^{\mathrm{H}} \mathbf{A}_k \mathbf{w} \\
&{\rm s.t.} \ \mathbf{w}^{\mathrm{H}} \mathbf{H}_q \mathbf{w} \leq \sigma_q , \  q=K+1,\dots,K+Q. 
\end{aligned}
\end{equation}

Note that beyond the aforementioned optical model, optimization problems involving quadratic forms frequently arise in various channel models and other related research domains~\cite{xiong2024optimaldiscrete}. We will develop methods tailored to effectively address this class of problems in the following Section. It is important to emphasize that the variable $\mathbf{w}$  resides in the union of tori~\footnote{Tori is the direct product of unit circles.}, a highly non-convex set, rather than the Hilbert space $\mathbb{C}^N$. Additionally, the domain of $\Omega_n$ ($n=1,\dots, N$) can be relaxed to $\mathbb{R}$ for inherent periodicity.

\section{Optimization Method for the Constrained max-min problem.}\label{Section3}
\subsection{Problem Reformulation and Method for Optimal Solution.}
Problem (P0) belongs to a class of semi-infinite Max-min problems with non-convex constraints. In what follows, we denote $\boldsymbol{\Omega} = [\Omega_1,\cdots,\Omega_N]^{\mathrm T}\in\mathbb{R}^N$ for simplified representation. Recall that $\mathbf{w} = [e^{j \Omega_1}, \cdots, e^{j \Omega_N}]^{\mathrm{H}} \in \mathbb{T}^N$ ($N$-tori), we define the follow functions
\begin{equation}
\begin{aligned}
\left\{
\begin{aligned}
    f_k(\mathbf{w})&= -\mathbf{w}^{\mathrm{H}} \mathbf{A}_k \mathbf{w}, \ \text{where}\ k =1,\ldots, K, \\
    g_q(\mathbf{w})&= \mathbf{w}^{\mathrm{H}} \mathbf{H}_q \mathbf{w}, \ \text{where}\ q =K+1,\ldots, K+Q.
    \end{aligned}
\right.
\end{aligned}
\end{equation}
Since $\mathbb{T}^N$ is compact, ${f}_k(\mathbf{w})$ is bounded and compact. Then the problem (P0) is equivalent to 
\begin{equation}\label{MM1}
\begin{aligned}
 \ \max_{ \boldsymbol{\Omega} } \min_{1\leq k\leq K} -f_k(\mathbf{w})  &=- \min_{ \boldsymbol{\Omega} } \max_{1\leq k\leq K} \ f_k(\mathbf{w})\\
\text {s.t.} \ g_q(\mathbf{w}) &\leq \sigma_q, 
\end{aligned}
\end{equation}
where $f_k(\mathbf{w})$ and $g_q(\mathbf{w})$ are continuous differentiable with Lipschitz continuous gradient (i.e., $C^{1,1}$ smooth). It is known that projection onto $\mathbb{R}^N$ is simple. As we introduce a variable $t \in \mathbb{R}$, the constrained minimization problem is equivalent to 
\begin{equation}\label{equation-t}
\begin{aligned}
\min _{\boldsymbol{\Omega}} \ t &\\
\text { s.t. } f_k(\mathbf{w})-t &\leq 0, \\
g_q(\mathbf{w})-\sigma_q &\leq 0.
\end{aligned}
\end{equation}

We define the function
\begin{equation}\label{min-max}
F(t)=\min _{\boldsymbol{\Omega}} \max \left\{\left\{ f_k(\mathbf{w})-t \right\}_{k=1}^K,\left\{ g_q(\mathbf{w})-\sigma_q \right\}_{q=K+1}^{K+Q} \right\} .
\end{equation}
Then, the original constrained Max-min problem \eqref{MM0} is equivalent to solving
\begin{equation}
F(t)=0.
\end{equation}
For simplicity of notations, we denote the vector-valued function 
\begin{equation}
\label{hKQ}
h_{K+Q}(\mathbf{w},t) = \left\{\left\{ f_k(\mathbf{w})-t \right\}_{k=1}^K,\left\{ g_q(\mathbf{w})-\sigma_q \right\}_{q=K+1}^{K+Q}\right\},
\end{equation}
which takes values in $\mathbb{R}^{K+Q}$.
The evaluation of $F(t)$ is equivalent to solving the unconstrained minimization problem
\begin{equation}\label{Ft}
F(t)=\min _{\boldsymbol{\Omega}}\; M_{K+Q}(h_{K+Q}(\mathbf{w},t)),
\end{equation}
where $M_{K+Q}(\cdot)$ is the maximum function in $\mathbb{R}^{K+Q}$.
We show the monotonicity of $F(t)$ in the following theorem.
\begin{thm}\label{thm:mono_de}
The function $F(t)$ is monotonically decreasing in $t$.
\end{thm}

\begin{proof}
We assume $t_1<t_2$ and $\mathbf{w}_1$ satisfying
\begin{equation}
F(t_1)= M_{K+Q}(h_{K+Q}(\mathbf{w}_1,t_1)).
\end{equation}
Then by \eqref{hKQ}\eqref{Ft} we have
\begin{equation}
\begin{aligned}
M_{K+Q}(h_{K+Q}(\mathbf{w}_1,t_1)) &> M_{K+Q}(h_{K+Q}(\mathbf{w}_1,t_2)) \\&> \min _{\boldsymbol{\Omega}}\; M_{K+Q}(h_{K+Q}(\mathbf{w},t_2)),
\end{aligned}
\end{equation}
which implies $F(t_1)>F(t_2)$.
\end{proof}

By monotonicity of $F$, $F(t)=0$ has a unique solution, and it can be found by the bisection method (linear convergence), in which only the sign of $F(t)$ is needed. The difficulty of evaluating $F(t)$ as defined in \eqref{min-max} lies in the non-smooth nature of the maximum function. We replace this function with a tight smooth approximation, given by compensated convexity (upper) transform~\cite{zhang2008compensated,zhang2022tight}, as
\begin{equation}\label{CCT}
C_\lambda^u\left(M_{K+Q}\right)(\mathbf{y})=\underbrace{ \lambda||\mathbf{y}||^2-\frac{||2\lambda\mathbf{y}-P_{\Delta_{K+Q} }(2\lambda\mathbf{y})||^2}{4\lambda}}_{\text{Moreau envelope}}+\frac{1}{4 \lambda}
\end{equation}
$\forall \ \mathbf{y}\in\mathbb{R}^{K+Q}$, where $\lambda$ is the regularization parameter, $\Delta_{K+Q}=\{\mathbf{x}\in\mathbb{R}_+^{K+Q}:\sum_{k=1}^{K+Q} x_{k} \leq 1\}$ is the unit simplex, and $P_\Delta$ is the projection onto $\Delta$. Note that this approximation differs from Moreau-Yosida approximation~\cite{yosida2012functional,xiong2024fair,moreau1965proximite} with only a constant gap of $\frac{1}{4\lambda}$~\cite{zhang2022tight}.
\begin{thm}\label{thm:error1}
There is a uniform error estimate for the above approximation, as
\begin{equation}
-\frac{1}{4\lambda}\leq M_{K+Q}(\mathbf{y})-C_\lambda^u\left(M_{K+Q}\right)(\mathbf{y})\leq 0 \text {. }
\end{equation}
\end{thm}

\begin{proof}
The proof can be found in \cite{zhang2022tight}.
\end{proof}

The tightness of the approximation is shown by the uniform convergence of the error as $\lambda \rightarrow \infty$. The gradient of the approximate function is given by
\begin{equation}
\nabla_{\mathbf{w}} C_\lambda^u\left(M_{K+Q}\right)(\mathbf{y})=P_{\Delta_{K+Q}}(2\lambda \mathbf{y}),
\end{equation}
which is $2 \lambda$-Lipschitz continuous.

Subsequently, a function that approximates $F(t)$ in~\eqref{min-max} can be defined as
\begin{equation}\label{uncons}
F_\lambda(t)=\min _{\boldsymbol{\Omega}} C_\lambda^u (M_{K+Q})(h_{K+Q}(\mathbf{w},t)) .
\end{equation}
The right-hand side (RHS) of~\eqref{uncons} is a smooth unconstrained optimization problem and can be solved by gradient-based methods such as Nesterov's accelerated gradient descent method \cite{xiong2024fair,nesterov1998introductory,devolder2014first}. Additionally, $F_\lambda(t)$ is also monotonically decreasing in $t$ and has a root for $\lambda$ sufficiently large. The following theorem gives the relation between $F(t)$ and $F_\lambda(t)$.
\begin{thm}\label{thm:error_est}
Suppose the minimizer of the RHS of~\eqref{min-max} is $\boldsymbol{\Omega}^*$, then the minimizer $\hat{\boldsymbol{\Omega}}$ of the RHS of~~\eqref{uncons} satisfies
\begin{equation}\label{eq:error_bnd}
-\frac{1}{4\lambda}\leq F(t)-F_\lambda(t)\leq 0.
\end{equation}
\end{thm}
\begin{proof}
Let ${\mathbf{w}}^*=e^{j\boldsymbol{\Omega}^*}$ and $\hat{\mathbf{w}}=e^{j\hat{\boldsymbol{\Omega}}}$. By Theorem~\ref{thm:error1} we have
\begin{equation}
\begin{aligned}
F_\lambda(t)=&C_\lambda^u(M_{K+Q})(h_{K+Q}(\hat{\mathbf{w}},t))
\geq M_{K+Q}(h_{K+Q}(\hat{\mathbf{w}},t))\\ \geq &M_{K+Q}(h_{K+Q}(\mathbf{w}^*,t)) = F(t).
\end{aligned}
\end{equation}
Similarly, we obtain
\begin{equation}
\begin{aligned}
F(t) = &M_{K+Q}(h_{K+Q}(\mathbf{w}^*,t)) \\
\geq &C_\lambda^u(M_{K+Q})(h_{K+Q}({\mathbf{w}^*},t)) - \frac{1}{4\lambda}
\\ \geq& C_\lambda^u(M_{K+Q})(h_{K+Q}(\hat{\mathbf{w}},t)) - \frac{1}{4\lambda} \\
=& F_\lambda(t) - \frac{1}{4\lambda}.
\end{aligned}
\end{equation}
This completes the proof.
\end{proof}
While employing the bisection method to find the root of $F$, $F_\lambda$ instead of $F$ is evaluated in each iteration. It satisfies that if $F_\lambda(t)<0 $, then $F(t)<0$, and if $F_\lambda(t)>\frac{1}{4 \lambda}$, then $F(t)>0$.
Definitely, $F_{\lambda}$ has a root for sufficiently large $\lambda$. However, we do not know for how large $\lambda$ such that $F_{\lambda}(t)$ has a root. Therefore, we adaptively increase $\lambda$ in our proposed algorithm. The algorithm to find the root of $F$ and the optimum of the original problem (P0) in~\eqref{MM0} is summarized in Algorithm~\ref{alg-1.0}.

\begin{remark} 
The proposed algorithm offers a general solution framework for constrained Max-min optimization problems. It operates by integrating the original objective function with the constraint conditions to construct a surrogate function and employs a bisection method for efficient solving. During the bisection optimization process, the algorithm evaluates an equivalent function $F(t)$, where the non-smooth maximum function in the Max-min problem is approximated using the compensated convexity (upper) transform as a proximal operator. This approximation enables the use of gradient-based methods to efficiently solve the resulting smooth optimization problem and obtain the optimal solution. 
\end{remark}

The proposed algorithm is applicable under the condition that the objective and constraint functions exhibit a similar mathematical structure. It is independent of specific signal models and the number of constraints, thereby enabling broad applicability to various constrained Max-min optimization problems.

\subsection{Convergence Analysis and Computational Complexity}

\begin{thm}\label{thm:error2}
Given any $\delta >0$, the output $t$ of Algorithm~\ref{alg-1.0} satisfies either $\left|t-t^* \right| \leq \epsilon $ or $\left| F(t)\right|\leq \delta$, where $t^*$ denotes the optimum in the RHS of~\eqref{min-max}.
\end{thm}

\begin{proof}
If Algorithm~\ref{alg-1.0} terminates at line 11, then 
\begin{equation}
0<F_\lambda(t)<\frac{1}{4\lambda}.
\end{equation}
By \eqref{eq:error_bnd} we have 
\begin{equation}
-\frac{1}{4\lambda}<F(t)<\frac{1}{4\lambda}.
\end{equation} 
Since $\lambda > \frac{1}{4\delta}$ at line 10, we have $-\delta<F(t)<\delta$.

If Algorithm~\ref{alg-1.0} concludes when $|a-b|<\epsilon$, then $t,t^*\in(a,b)$ and $|t-t^*|<\epsilon$.
\end{proof}

\begin{remark} In case $\left| F(t)\right|\leq \delta$ at termination of Algorithm~\ref{alg-1.0}, the maximum violation of the constraints \{$g_q(\mathbf{w}) -\alpha_q\leq 0\}_{q=K+1}^{K+Q}$ is no more than $\delta$. Hence $\delta$ serves as the constraint tolerance for solving the constrained Max-min problem~\eqref{MM0}.
\end{remark}
\begin{algorithm}[!tbp]
\caption{Bisection-based method (BIS) for constrained Max-min problem~\eqref{MM0}}
\label{alg-1.0}
\begin{algorithmic}[1]
\Require $\mathbf{A}_k$, $ \mathbf{H}_q$, the threshold $\sigma_q$, initial parameter $\lambda_0$, tolerance $\epsilon$, $\delta$. 
\Ensure Optimum $\mathbf{\Omega}^*$, the root $t$  
\State Construct $F_\lambda(t)$ as in~\eqref{uncons} and initialize $a,b$, such that $F_{\lambda}(a) > 0$ and $F_\lambda(b) < -\frac{1}{4\lambda}$;
\State $\lambda \gets \lambda_0$;
\While{$|a - b| > \epsilon$}
    \State $t \gets a + \frac{b - a}{2}$;
    \State Solve the RHS of~\eqref{uncons} for optimum $\mathbf{\Omega}^*$ and evaluate $F_{\lambda}(t)$;
    \If{$F_{\lambda}(t) < -\frac{1}{4\lambda}$}
        \State $b \gets t$;
    \ElsIf{$F_{\lambda}(t) > 0$}
        \State $a \gets t$;
    \ElsIf{$\lambda > \frac{1}{4\delta}$}
     \State   \textbf{break};
    \Else
        \State $\lambda \gets 2\lambda$;
    \EndIf
\EndWhile
\State \textbf{return} $\mathbf{\Omega}^*$, $t$
\end{algorithmic}
\end{algorithm}

The computational complexity of the bisection search is $\mathcal{O}(\log L)$, where $L = a-b$ represents the size of the search space. As the algorithm converges, the functional value $t$ and corresponding optimal solution $\mathbf{\Omega}^*$ for the original constrained Max-min problem~\eqref{MM0} can be achieved.

\section{Numerical Simulation Evaluations}\label{Section4}
We now investigate the beam synthesis capabilities of transmissive RIS to demonstrate the effectiveness of the proposed Max-min framework. Initially, we validate the explicit beam suppression functionality, which is most overlooked in RIS-assisted wireless communications. Subsequently, we thoroughly evaluate the flexible multi-beam synthesis performance and analyze the impact of the beam constraint threshold constant on solution behavior. Finally, we compare the proposed method with two established algorithms: SDR-SDP~\cite{li2024weighted} and QuantRand~\cite{subhash2023max}, regarding beam pattern, received signal power, and computational efficiency. 

The SDR-SDP method is related to semi-definite relaxation and semi-definite programming techniques~\cite{wu2019intelligent}, it relaxes the uni-modular constraint on the transmission coefficient to rank-one restrictions on the transformed variable matrix. Thus, the problem becomes convex and can be solved using an interior point method. QuantRand is a quantized variant of random coordinate descent algorithm, employing a greedy search strategy.
For a comprehensive evaluation, an unconstraint beam synthesis approach is also considered as a benchmark, referred to as Non-Constraint. This method utilizes the MA algorithm~\cite{xiong2024fair}, which has already demonstrated its potential in fair beam allocation, to solve problem (P0) without imposing any directional suppression constraints. In this case, multi-beam synthesis is solely directed towards communication UEs, without any consideration of sidelobe control.
Unless otherwise specified, the general simulation parameters are summarized in Table~\ref{Common} to avoid redundant descriptions.

\begin{table}[tbp]
\caption{General Simulation Setup Parameters} \label{common}
\centering
\setlength{\tabcolsep}{1.2mm}
\setstretch{1.1} 
\begin{tabular}{ccc}
\toprule[1.5pt]
Description &  Symbol & Value  \\
\midrule
Operating frequency & $f$    & 5.8 GHz \\
Wavelength   &$\lambda$  & 0.0517 m\\
\multirowcell{2}{The number of units\\} &\multirowcell{2}{$N$\\}  & \multirowcell{2}{$1\times16$ (linear)\\$16\times16$ (planar) } \\ \\ 
Inter-unit spacing  &$d $  &  0.025 m  \\
Tx1 position  &($r^\text{t}_1,\theta^{\text t}_1, \phi^{\text t}_1$) &(5 m, $0^{\circ},0^{\circ}$)\\ 
Tx2 position   &($r^\text{t}_2,\theta^{\text t}_2, \phi^{\text t}_2$) & (5 m, $30^{\circ},120^{\circ}$)\\
Tx3 position   &($r^\text{t}_3,\theta^{\text t}_3, \phi^{\text t}_3$) &  (5 m, $ 40^{\circ},240^{\circ}$)\\
Transmit power at each Tx& $P_{\rm t}$   & 0 dBm \\
UE Distance &$r^{\text r}$ &30 m\\ 
\bottomrule[1.5pt] 
\label{Common}
\end{tabular}
\end{table}

\subsection{Directional Suppression}

We begin by evaluating the beam suppression performance using a linear transmissive RIS consisting of 16 units. For simplified notation, we denote the elevation angle range from $-90^{\circ}$ to $90^{\circ}$ in such two-dimensional observation planes, which applies similarly in the following sections. A single transmitter (Tx1) is employed, and two communication UEs are positioned at directions $-20^{\circ}$ and $30^{\circ}$ with equal weight factors. Additionally, a potential eavesdropping or interference threat is assumed within the angular interval $[-10^{\circ}, 0^{\circ}]$. The cross-sections of the beams at distance $r^\text{r} = 30$ m, as a function of the observation angle $\theta$, are illustrated in Fig.~\ref{BeamSuppress}.
\begin{figure}[!htbp]
  \centering
  \subfigure[]{
  \label{beamsuppressed}
  \includegraphics[width=0.7\linewidth]{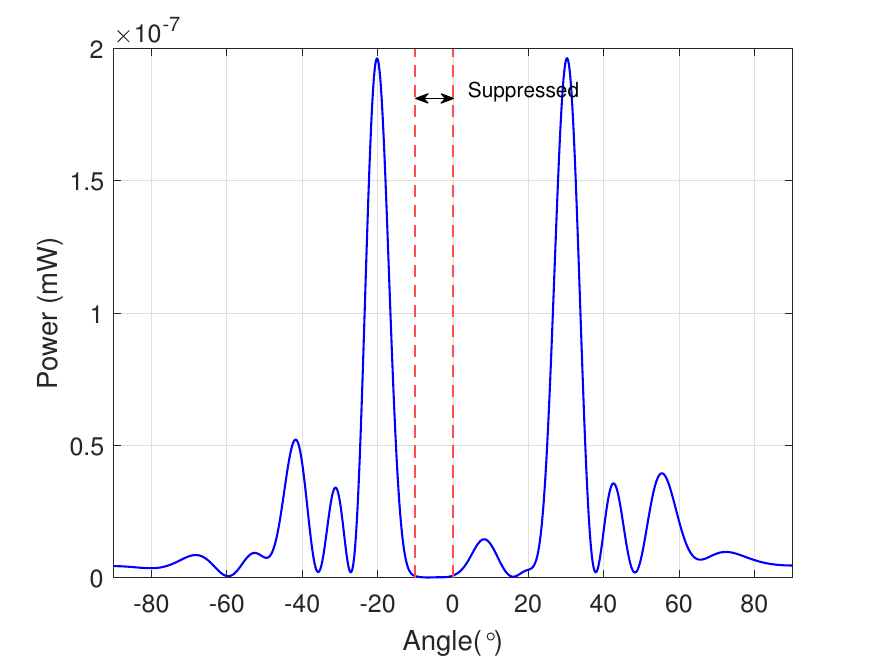}}
  \subfigure[]{
  \label{beamsuppressedpolar}
  \includegraphics[width=0.8\linewidth]{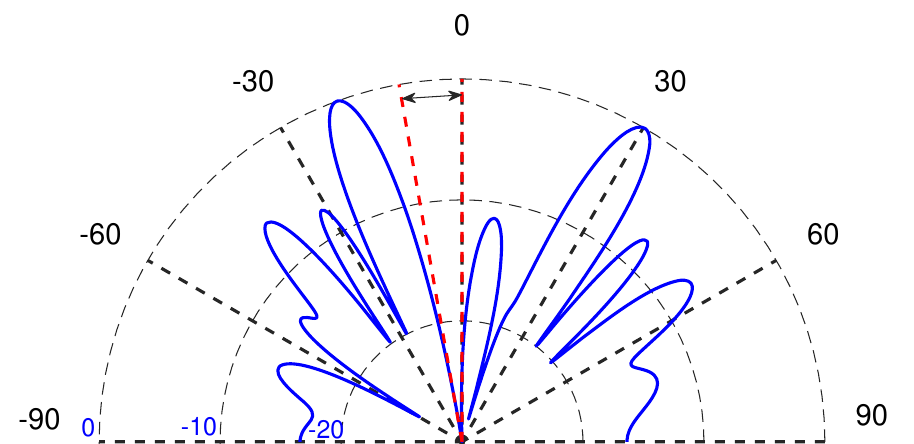}}
  \caption{The radiation pattern for beam synthesis illustration. A single signal source (Tx1) is employed; the observation UE angle is set to be $-20^{\circ}$ and $30^{\circ}$. The suppression interval is $[-10^{\circ},0^{\circ}]$, bounded by two red dashed lines. (a) The power plot. (b) The normalized polar power plot.} 
  \label{BeamSuppress}
\end{figure}
The results in Fig.~\ref{beamsuppressed} indicate that, in addition to forming high-gain directional beams at the target UE directions, the signal energy within the interval from $-10^{\circ}$ to $0^{\circ}$ is completely suppressed as intended. This suppression effect is more clearly observed in the normalized polar plot shown in Fig.~\ref{beamsuppressedpolar}, where a distinct beam null appears in the suppression region.

Furthermore, a planar RIS array and three sources, each positioned at a uniform distance of 5 meters and with orientations of ($0^{\circ}, 0^{\circ}$), ($30^{\circ}, 120^{\circ}$), and ($40^{\circ}, 240^{\circ}$), are employed to repeat the experiment. In this test, the directions of the two target communication UEs are adjusted to ($120^{\circ}, 0^{\circ}$) and ($160^{\circ}, 180^{\circ}$), with a separation of 30 meters from the RIS. We plot the normalized three-dimensional beam radiation pattern. The results in Fig.~\ref{16-16CON-3D} indicate that, in contrast to the unconstraint beam synthesis shown in Fig.~\ref{16-16NON-CON-3D}, the proposed approach effectively controlled and suppressed the directional beams that could potentially refer to unauthorized users. It should be noted that, due to energy conservation, achieving beam suppression in all sidelobe directions simultaneously is not feasible.

\begin{figure}[tbp]
  \centering
  \subfigure[]{
  \label{16-16NON-CON-3D}
  \includegraphics[width=0.475\linewidth]{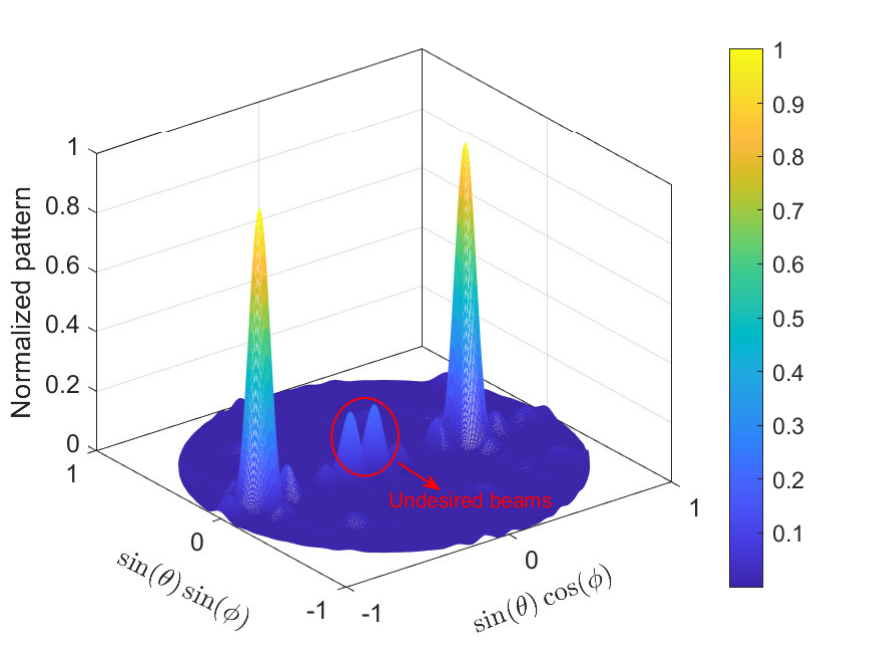}}
  \subfigure[]{
  \label{16-16CON-3D}
  \includegraphics[width=0.475\linewidth]{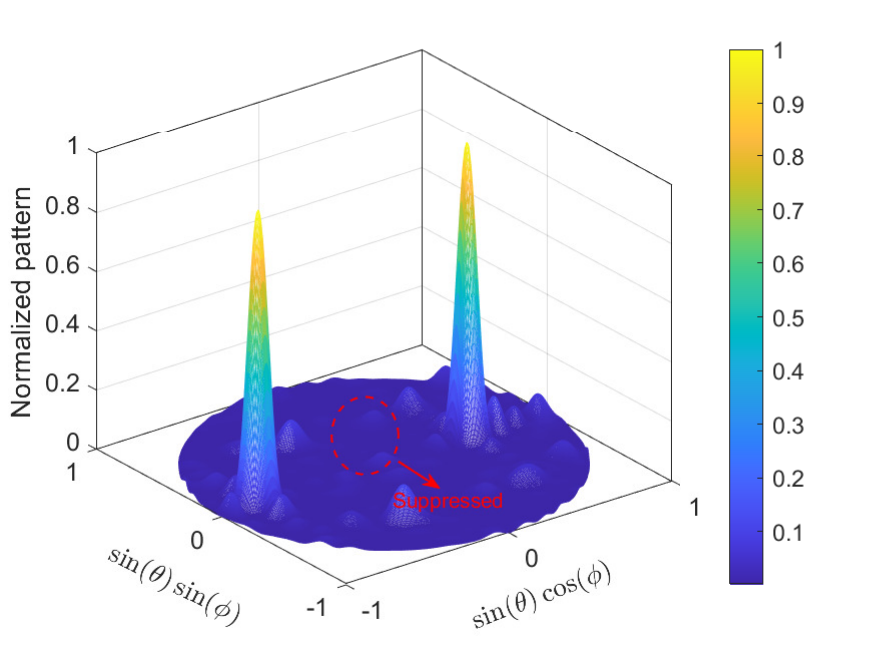}}
  \caption{The normalized 3-dimensional (3D) beam radiation pattern. (a) The unconstraint radiation pattern generated by the MA algorithm~\cite{xiong2024fair}. (b) The sidelobe beams are suppressed within the elevation angle range of $[0^{\circ}, 20^{\circ}]$ and the azimuth angle range of $[179^{\circ}, 181^{\circ}]$ through the proposed method.}
  \label{3D}
\end{figure}
\begin{figure}[!htbp]
  \centering
  \subfigure[]{
  \label{BeamsplitingPolar}
  \includegraphics[width=0.78\linewidth]{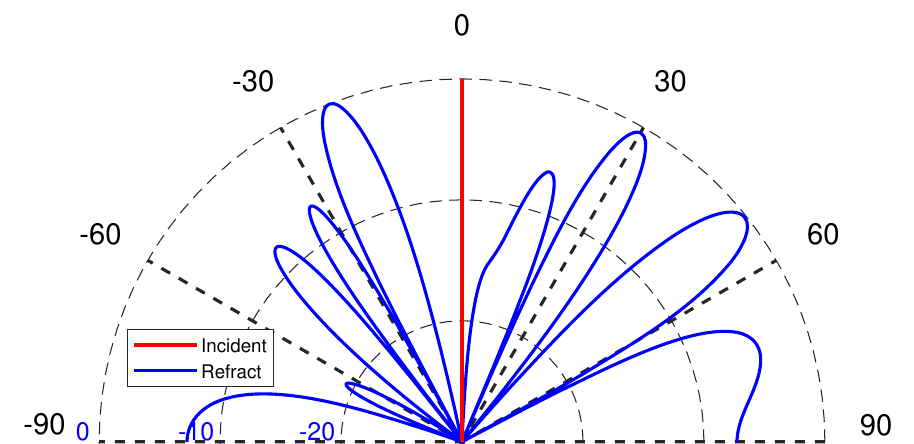}}
  \subfigure[]{
  \label{BeamAggregation-Polar}
  \includegraphics[width=0.78\linewidth]{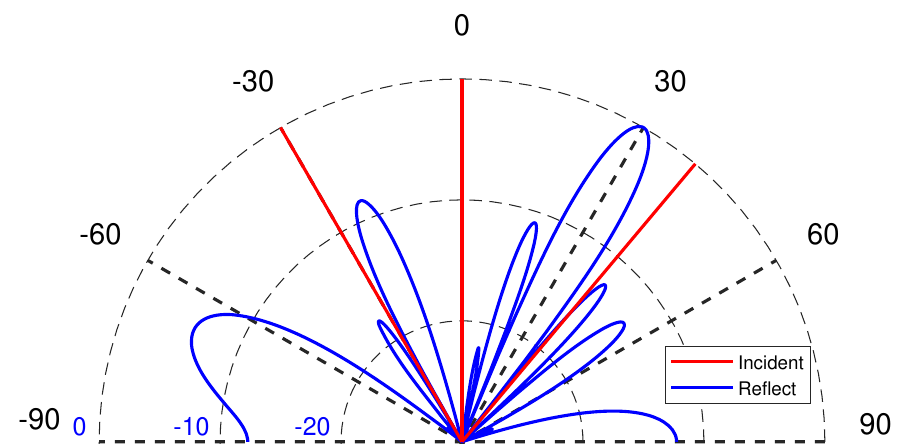}}
  \caption{Beam splitting and aggregation under the constraint of beam suppression within the angular range of $[-10^\circ, 0^\circ]$. (a) A single incident signal is split for three directional UEs. (b) Three incident signals are aggregated into a focused beam.}
  \label{BeamFunc}
\end{figure}

\subsection{Flexible Beam Synthesis}

\subsubsection{Beam-spliting and beam-aggregation}
Building upon the directional suppression, we further evaluate the flexible beam synthesis capability including beam-splitting and beam-aggregation, as shown in Fig.~\ref{BeamFunc}. While maintaining the beam null within the angular range of $-10^{\circ}$ to $0^{\circ}$, the proposed method is capable of various forms of beam manipulation, including splitting a single incident beam into multiple beams for multi-user communication and aggregating multiple incident beams into a single focused beam to achieve high directional signal enhancement. Additionally, the proposed method can also be utilized to implement widebeam through synthesizing multiple sub-beams, as discussed in our previous work~\cite{xiong2024fair}. Furthermore, the beamwidth narrows as the number of units increases, as we have also discussed. 

\subsubsection{Flexible beam synthesis}
In addition to beam suppression, the proposed framework also enables flexible beam control across different UE directions by adjusting the weight factors in~\eqref{MM0}. As demonstrated in Fig.~\ref{BeamSuppress-12} and Fig.~\ref{BeamSuppress-124}, while maintaining the beam null within the designated suppression region of $[-10^\circ, 0^\circ]$, we successfully acheive beam energy synthesis in a 2:1 ratio for two UE directions at $\{-20^{\circ}, 20^{\circ}\}$, and in a 4:2:1 ratio for three different UE directions at $\{-50^{\circ}, -20^{\circ},0^{\circ}\}$. This functionality allows us to precisely control the desired spatial energy distribution, whether for signal enhancement or suppression.


\begin{figure}[!htbp]
  \centering
  \subfigure[]{
  \label{BeamSuppressed-12}
  \includegraphics[width=0.68\linewidth]{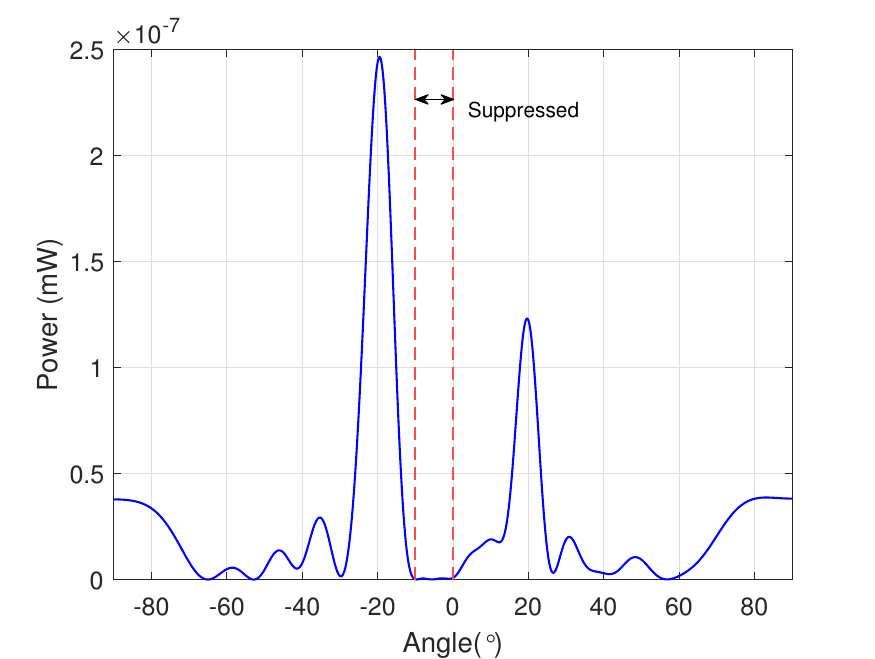}}
  \subfigure[]{
  \label{BeamSuppressed-Polar-12}
  \includegraphics[width=0.78\linewidth]{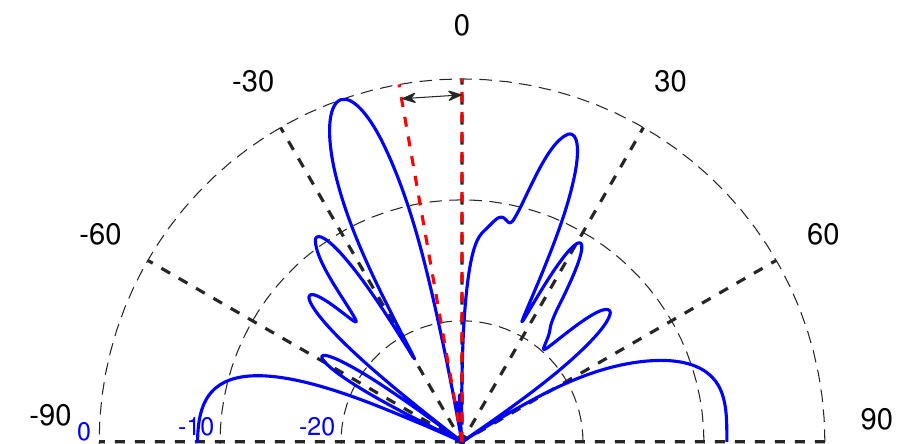}}
  \caption{Flexible beam synthesis with a single incident signal. The suppression interval bounded by two red dashed lines is $[-10^{\circ},0^{\circ}]$. UE $1$ is positioned at $-20^{\circ}$, and UE $2$ is positioned at $20^{\circ}$, with corresponding weight factors $\alpha_1 = 2,\ \alpha_2 = 1$. (a) The radiation pattern. (b) The normalized radiation pattern in a polar plot.}
  \label{BeamSuppress-12}
\end{figure}

\begin{figure}[!htbp]
  \centering
  \includegraphics[width=0.85\linewidth]{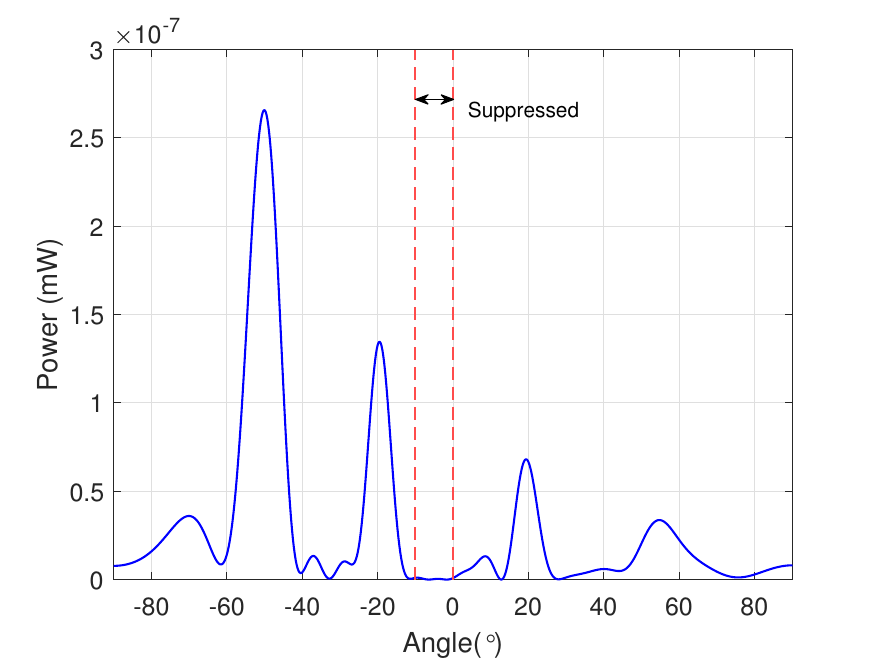}
  \caption{Flexible beam synthesis with a single incident signal. The three communication UEs are positioned at directions $-50^{\circ}$, $-20^{\circ}$, and $20^{\circ}$, with corresponding weight factors of $\alpha_1 = 4$, $\alpha_2 = 2$, and $\alpha_3 = 1$, respectively.}
  \label{BeamSuppress-124}
\end{figure}

\subsection{The Impact of Constraint Threshold on Beam Pattern}
In the process of beam synthesis, the amplitude of the threshold $\sigma_q$ in problem (P0) impacts system performance by altering the size of the feasible space, as we discussed in Section~\ref{Section2}. To evaluate this effect, we conduct a comparison experiment. Specifically, we first obtain the radiation pattern of the linear RIS without sidelobe constraints and record the power values at two observation positions with equal weight, labeling them as `$\text{Peak}$'. Then, we assign different $\sigma_q$ values based on varying proportions of `$\text{Peak}$' in the suppression region, and observe the resulting impact on the radiation pattern. The suppression interval here is set to $[-34.5^{\circ},-19.5^{\circ}]$.

The results in Fig.~\ref{CT} indicate that smaller threshold $\sigma_q$ values correspond to a more limited feasible solution space, leading to more pronounced suppression of sidelobe beams. In the target sidelobe region, the suppression effect with $\sigma_q = 0.01 \times \text{Peak}$ shows an average gain of 20 dB compared to $\sigma_q = 0.5 \times \text{Peak}$, and achieves 25 dB gains compared to the Non-Constraint scheme. It is worth noting that the introduction of the sidelobe constraint (related to the unauthorized UE direction) results in a slight reduction of the received signal for the communication UEs in the main beam direction. Compared to the Non-Constraint case, the power loss for the target communication UEs with $\sigma_q = 0.01 \times \text{Peak}$ is approximately 0.7 dB. These comprehensive results demonstrate that the proposed framework can achieve substantial sidelobe suppression gains with only a marginal reduction in the main beam.

\begin{figure}[!htbp]
  \centering
  \subfigure[]{
  \label{CompareThreshold}
  \includegraphics[width=0.75\linewidth]{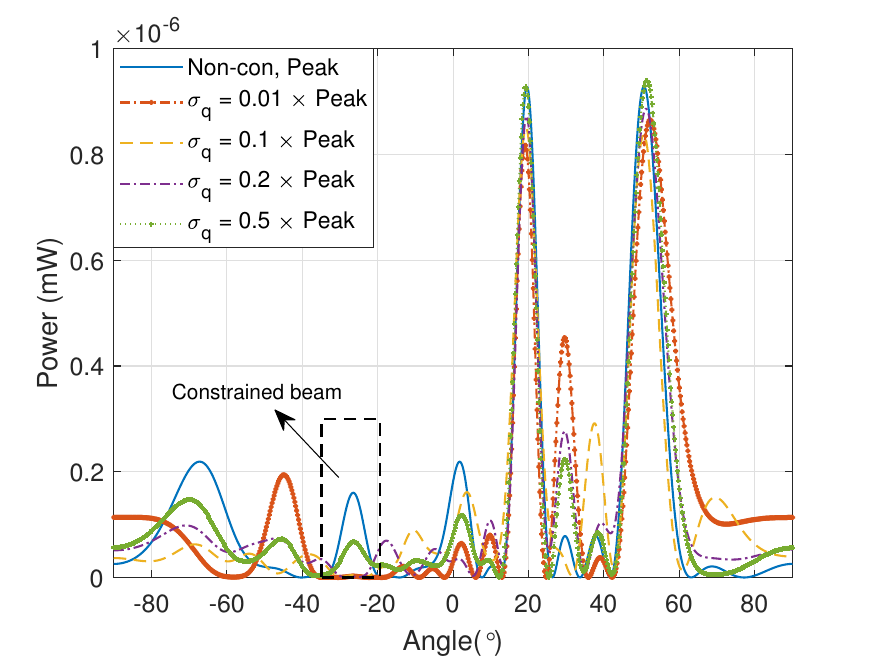}}
  \subfigure[]{
  \label{Compare-Threshold-dB}
  \includegraphics[width=0.75\linewidth]{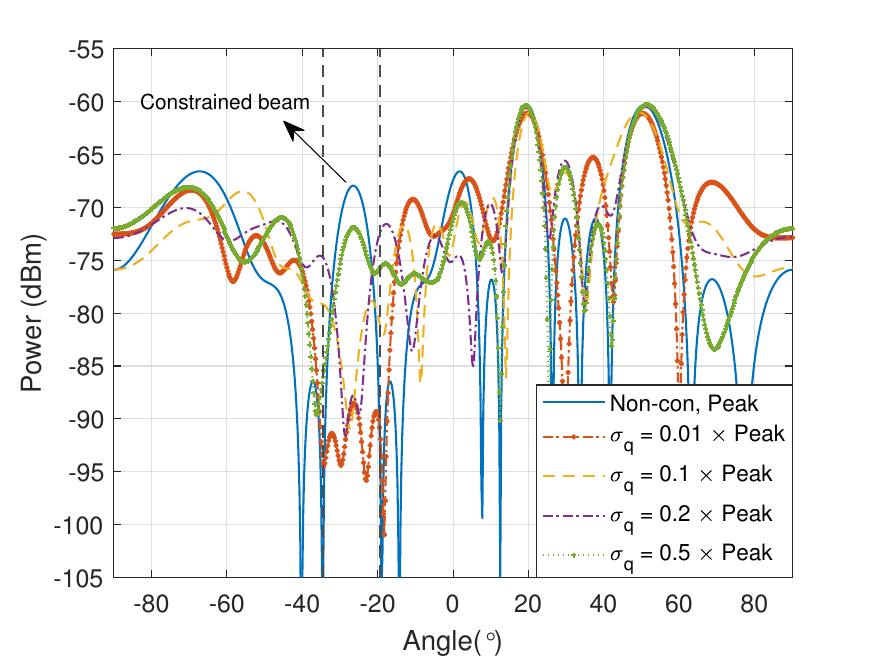}}
  \caption{Impact of different $\sigma_q$, while `Non-con' is related to the Non-Constraint scheme. Three transmitters and a linear RIS are employed, and the target UE directions are $20^{\circ}$ and $50^{\circ}$. The suppressed region is set to $[-36^{\circ},-20^{\circ}]$.}
  \label{CT}
\end{figure}

\subsection{Beam Performance Comparison}\label{4-Beam}
To further evaluate the beam performance of the proposed algorithm, we compare it against two existing methods, including SDR-SDP~\cite{li2024weighted} and QuantRand~\cite{subhash2023max} methods. Additionally, we present the outcomes of the Non-Constraint case as a baseline. In this experiment, the linear transmissive RIS employed is illustrated by a single source, Tx1. The suppressed region is set from $-36^{\circ}$ to $-20^{\circ}$, and the weight factors for the two served communication UEs at $20^{\circ}$ and $50^{\circ}$ are set to be equal. We present the cross-sections of the beams at a distance $r^\text{r} = 30$ m, as a function of the observation angle $\theta$, as shown in Fig.~\ref{CompareBeampattern}.

It can be observed that the proposed BIS algorithm achieves effective beam energy suppression in the selected region while maintaining a main beam power level comparable to the Non-Constraint baseline. In contrast, the SDR-SDP method not only falls short in terms of main beam power level but also delivers suboptimal suppression results. Furthermore, while the QuantRand algorithm exhibits commendable suppression performance to some extent, its poor performance in the energy distribution of the two main beams indicates that it is not a reasonable solution. 

For completeness, we vary the weight factors and repeat the tests, with the corresponding results and performance metrics associated with various algorithms summarized in Table~\ref{Metrics}. It can be concluded that the proposed BIS algorithm exhibits superior performance in beam suppression and controlling the signal power ratio compared to existing methods. These results demonstrate that the proposed algorithm simultaneously achieves both the main beam control and directional sidelobe suppression, which none of the existing work can accomplish.

\begin{figure}[htbp]
  \centering
  \includegraphics[width=0.85\linewidth]{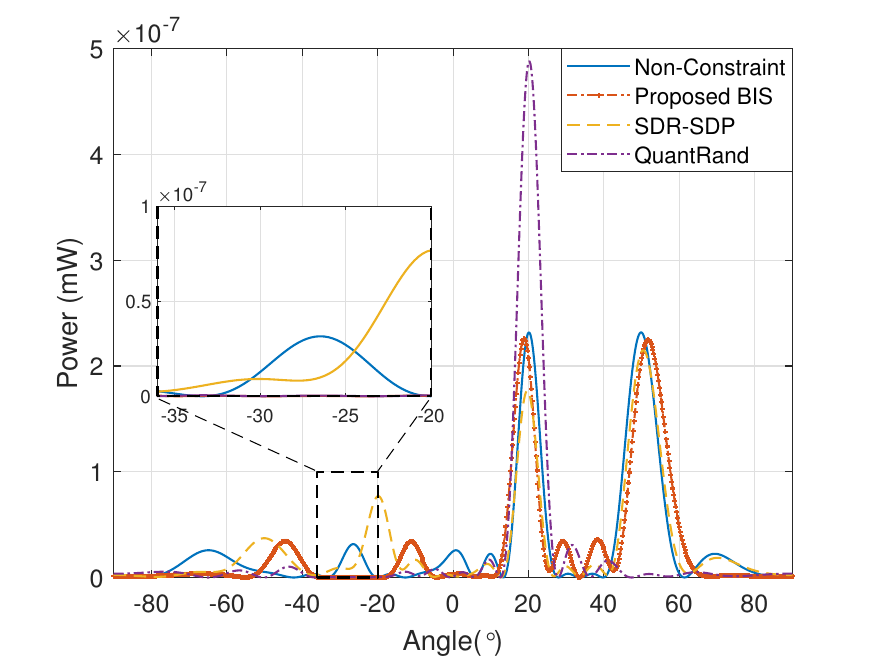}
  \caption{Beampattern comparison of various methods. A single transmitter is employed. The suppressed region is set to $[-36^{\circ},-20^{\circ}]$.}
  \label{CompareBeampattern}
\end{figure}

\begin{table*}[!htbp]
\caption{Metrics Comparisons, the Suppressed angular interval is ($\theta_{\text {sup}}$ = $[-36^{\circ},-20^{\circ}]$)} \label{Variance}
\centering
\setlength{\tabcolsep}{1.2mm}
\setstretch{1.1} 
\begin{tabular}{ccccc}
\toprule[1.5pt]
Weight \& Constraint& Methods &Actual power ratio  & Variance & Peak power within the suppressions \\
\multirowcell{4}{$\alpha_1:\alpha_2=1:1$ \\(power ratio 1:1)}& \multirowcell{4}{Non-Contraint\\ \textbf{Proposed BIS} \\ QuantRand \\SDR-SDP   } & \multirowcell{4}{ 1:1\\ \textbf{1:1}\\64.21:1\\1.62:1 }&  \multirowcell{4}{0 \\ \textbf{0}\\1997.5 \\0.1947} & \multirowcell{4}{ -75.01 dBm\\\textbf{-82.06} dBm\\ -84.06 dBm\\-70.12 dBm}      \\ \\ \\ \\
\midrule
\multirowcell{4}{$\alpha_1:\alpha_2=1:2$  \\(power ratio 1:2)}& \multirowcell{4}{Non-Contraint\\ \textbf{Proposed BIS} \\ QuantRand \\SDR-SDP   } & \multirowcell{4}{1:2\\ \textbf{1:2} \\163.96:2\\1:2.19 }&  \multirowcell{4}{0 \\ \textbf{0}\\13277.98 \\0.0181} & \multirowcell{4}{-75.6266 dBm\\ \textbf{-94.6233} dBm\\-92.4337 dBm\\-72.8209 dBm}      \\ \\ \\ \\
\bottomrule[1.5pt] 
\label{Metrics}
\end{tabular}
\end{table*}

\subsection{Evaluation on the Received Signal Power}
We proceed to analyze the received signal power at the target communication UEs and the observation points (unauthorized UEs) within the suppression region. In this context, two primary parameters are considered. The first is the minimum received signal power among all communication UEs (Min-UE), which serves as a critical metric for evaluating the effectiveness of different algorithms in ensuring overall communication quality in the system. The second parameter is the maximum received signal power at the sample points within the suppression region (Max-SP), used to verify whether the imposed sidelobe power constraints are satisfied. This evaluation employs a planar RIS array and three Tx as described in Table~\ref{Common}, with three communication UEs, each assigned equal weight factors, located at $(20^{\circ},0^{\circ}),(50^{\circ},0^{\circ})$, and $(50^{\circ},180^{\circ})$. The suppression region is defined as $\theta = [20^{\circ}, 36^{\circ}]$, $\phi = [179^{\circ}, 181^{\circ}]$. All observation points and UEs are uniformly positioned at a distance of 30 m from the RIS.  
\begin{figure}[!htbp]
  \centering
  \subfigure[]{
  \label{CDF-UE-maximum}
  \includegraphics[width=0.75\linewidth]{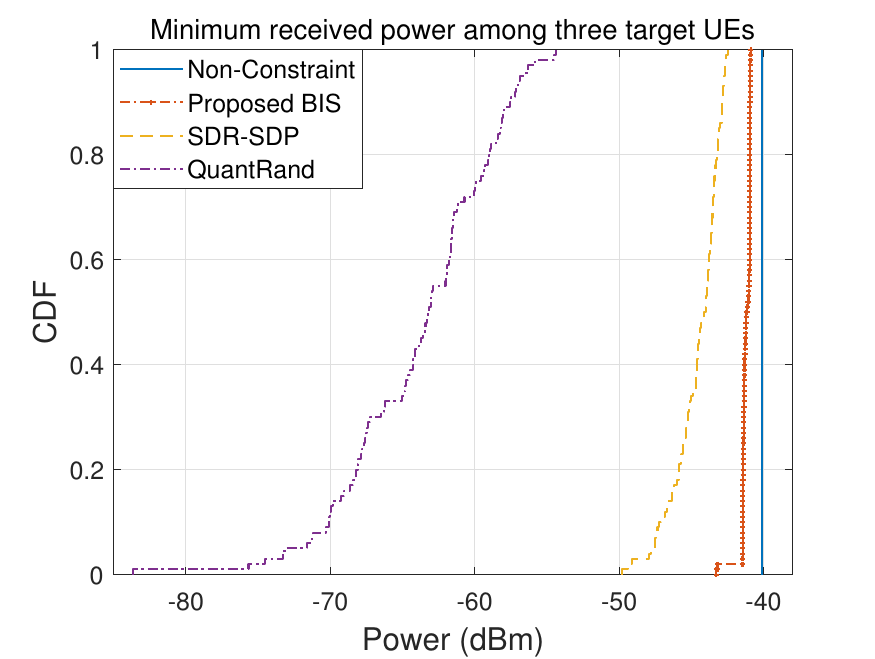}}
  \subfigure[]{
  \label{CDF-SUP-minimum}
  \includegraphics[width=0.75\linewidth]{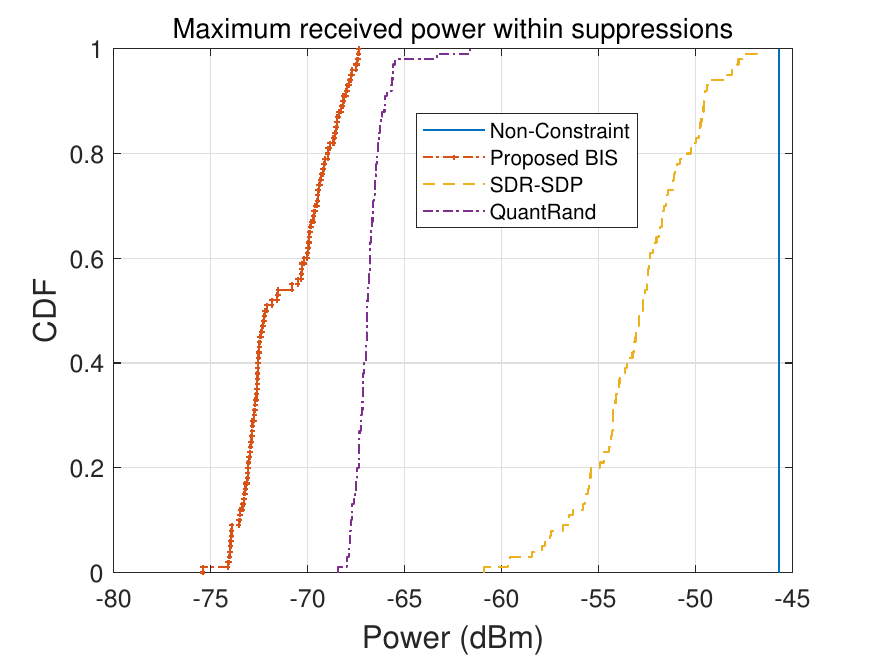}}
  \caption{Cumulative Distribution Function of the received signal power at different observation points. (a) The minimum across all target UEs (Min-UE). (b) The maximum of all suppressed points (Max-SP).}
  \label{CDF}
\end{figure}

As shown in Fig.~\ref{CDF}, the cumulative distribution function (CDF) illustrates that the proposed BIS outperforms the other two algorithms in both signal enhancement and suppression performance. Specifically, with reference to the Non-Constraint baseline, the BIS algorithm achieves a reduction of approximately 27 dB in the maximization power, across all assumed unauthorized UEs within the suppression region, while incurring only a 1.2 dB loss in the power received by the target communication UEs. In contrast, the SDR-SDP method exhibits a loss and gain of 4.4 dB and 7 dB, respectively, while the QuantRand algorithm incurs a power loss of 22 dB with a suppression gain of only 11 dB.

\subsection{Analysis of the Relative Gain Performance}

For a clearer demonstration of the performance gains achieved by the proposed methods, we define a metric referred to as Relative Gain (RG), calculated by 
\begin{equation*}
\text{RG} = 10 \log_{10} (P_{\text{UE},min}/P_{\text{SP},max}),
\end{equation*}
where $P_{\text{UE}, min}$ denotes the minimum received power among all communication UEs, and $P_{\text{SP}, max}$ represents the maximum power at the observation points associated with potential unauthorized UEs. This metric offers insights for analyzing the proposed algorithm's performance in physical layer security and interference mitigation applications.

We conduct two sets of experiments: the first focusing on the impact of the number of RIS units on RG performance, while the second examines the influence of the number of users. In the first test, we keep the three communication UEs at $(20^{\circ}, 0^{\circ})$, $(50^{\circ}, 0^{\circ})$, and $(50^{\circ}, 180^{\circ})$, with the suppressed region set to $[-36^{\circ}, -20^{\circ}]$. The number of units $N$ on the linear RIS is incrementally increased. In the second test, the number of RIS units is fixed at $N = 16$, while the number of UEs $K$ is gradually increased. The RG performance of various methods is evaluated across these scenarios. To ensure robustness, we conduct independent 100 trials for each value of $N$ or $K$ and calculate the average values. 
\begin{figure}[!htbp]
  \centering
  \subfigure[]{
  \label{RG-N}
  \includegraphics[width=0.75\linewidth]{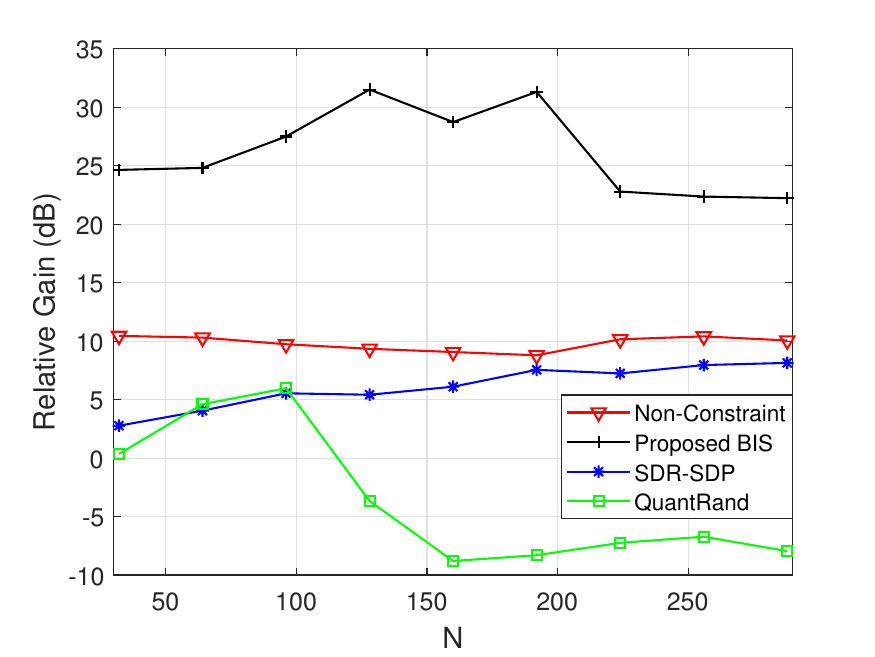}}
  \subfigure[]{
  \label{RG-K}
  \includegraphics[width=0.75\linewidth]{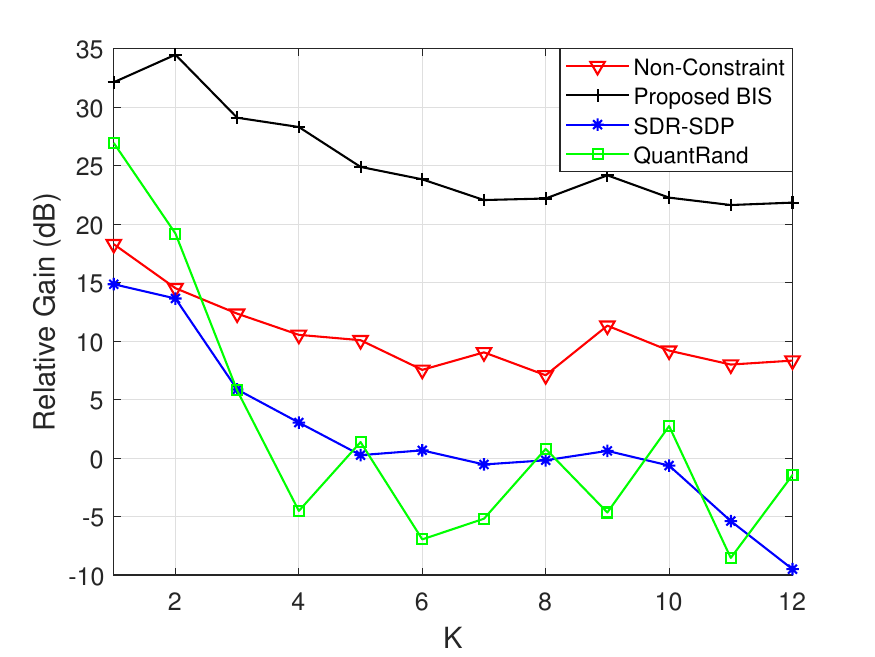}}
  \caption{Comparison across different methods. Three Tx are employed, the suppressed angular range is set to $[-36^{\circ},-20^{\circ}]$. (a) RG performance versus the number of RIS units. (b) RG performance versus the number of UEs.}
  \label{Compare}
\end{figure}


The results in Fig.~\ref{RG-N} indicate that, when using the Non-Constraint scheme, which focuses on multi-user power allocation without considering sidelobe energy constraints, as a baseline, the RG performance of this standard scheme remains relatively stable at 10 dB, regardless of the increase in the number of units. This suggests a dynamic equilibrium between the mainbeam and sidelobe energy: as the number of RIS units grows, both the mainbeam power level and sidelobe energy increase proportionally. Consequently, while deploying more RIS units can enhance signal and communication quality in multi-user scenarios, it also introduces the potential for increased interference or information leakage in unintended directions. This trade-off highlights the importance of incorporating sidelobe constraints in RIS optimization to mitigate adverse effects in practical applications.

Focusing on sidelobe energy, we observe that SDR-SDP and QuantRand exhibit even poorer RG performance than Non-Constraint, particularly when the number of units exceeds 120. In this context, the RG achieved by the QuantRand algorithm falls below zero, indicating that the beam energy directed towards the interference/eavesdropping direction exceeds that of the intended communication UEs.
In contrast, the proposed BIS algorithm achieves an RG performance of up to 25 dB. This substantial difference between the main beam and the suppressed sidelobe power level ensures its capability to prevent signal interference and information leakage in large-scale RIS deployments. 

A similar trend is observed when varying the number of users $K$, as depicted in Fig.~\ref{RG-K}. The proposed method's RG performance consistently surpasses that of the comparative algorithms. Notably, as the number of users increases, a gradual decline in RG is observed across all methods. This is due to the equal-weight allocation scheme, wherein the signal power received by each user decreases proportionally as the total number of users increases. Nevertheless, the proposed method still achieves an RG of 22 dB when the number of communication UEs reaches 12, demonstrating its remarkable robustness.

It is noteworthy that the performance curves of certain baseline algorithms, such as QuantRand, exhibit non-monotonic fluctuations as the number of users $K$ increases. This phenomenon may be attributed to: 1) the presence of significant high-variance characteristics during the optimization process, , as discussed in relation to the results in Table~\ref{Variance}; and 2) the inherent sensitivity of the RG metric to variations in user count, wherein the algorithm's robustness becomes increasingly influenced by the intermediate quantity $P_{\text{UE}, min}$ as $K$ grows. These findings substantiate the superior performance and stability of the proposed method in multi-user scenarios.

\begin{figure}[htbp]
  \centering
  \includegraphics[width=0.85\linewidth]{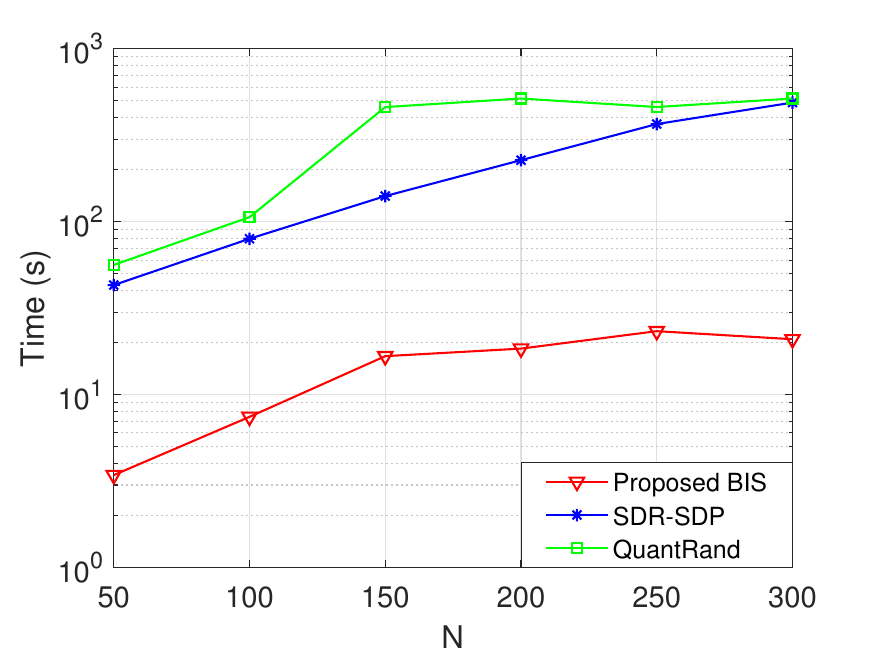}
  \caption{The overall processing time comparison across various algorithms.}
  \label{5-Time-N}
\end{figure}

\subsection{Time Complexity}

To evaluate the computational efficiency of the proposed BIS algorithm, this section tests the runtime performance of different methods on a computing platform equipped with an ${\rm \text{Intel}^{\text R} \text{Core}^{\text{TM}}}$ i7-14700KF processor (3.40 GHz, 64 GB RAM) running MATLAB R2023a. The evaluation considers varying numbers of RIS units $N$. For each value of $N$, 20 trials are conducted, and the average runtime is recorded and plotted. The results shown in Fig.~\ref{5-Time-N} indicate that the computational complexity of all methods increases with the number of units. Compared with the existing SDR-SDP and QuantRand approaches, the proposed BIS method consistently demonstrates the lowest overall computational complexity across different values of $N$.

\section{Prototype Measurements}\label{Section5}
To evaluate the effectiveness of the proposed method in real-world environments, we conduct comparative experiments in an anechoic chamber using a 1-bit RIS composed of $16 \times 16$ units~\cite{zhang2024design}. The experimental setup is illustrated in Fig.~\ref{Darkroom}, where the transmissive RIS and the transmitting horn antenna are co-mounted on a rotating platform, with the horn antenna positioned 0.6 m directly behind the RIS. The electromagnetic wave propagates through the transmissive RIS, is redirected, and subsequently received by another horn antenna (Rx) fixed at a distance of 5.0 m from the RIS. The Rx horn is connected to a vector network analyzer, which measures and records the received signal power within the horizontal plane.
\begin{figure}[htbp]
  \centering
  \includegraphics[width=0.95\linewidth]{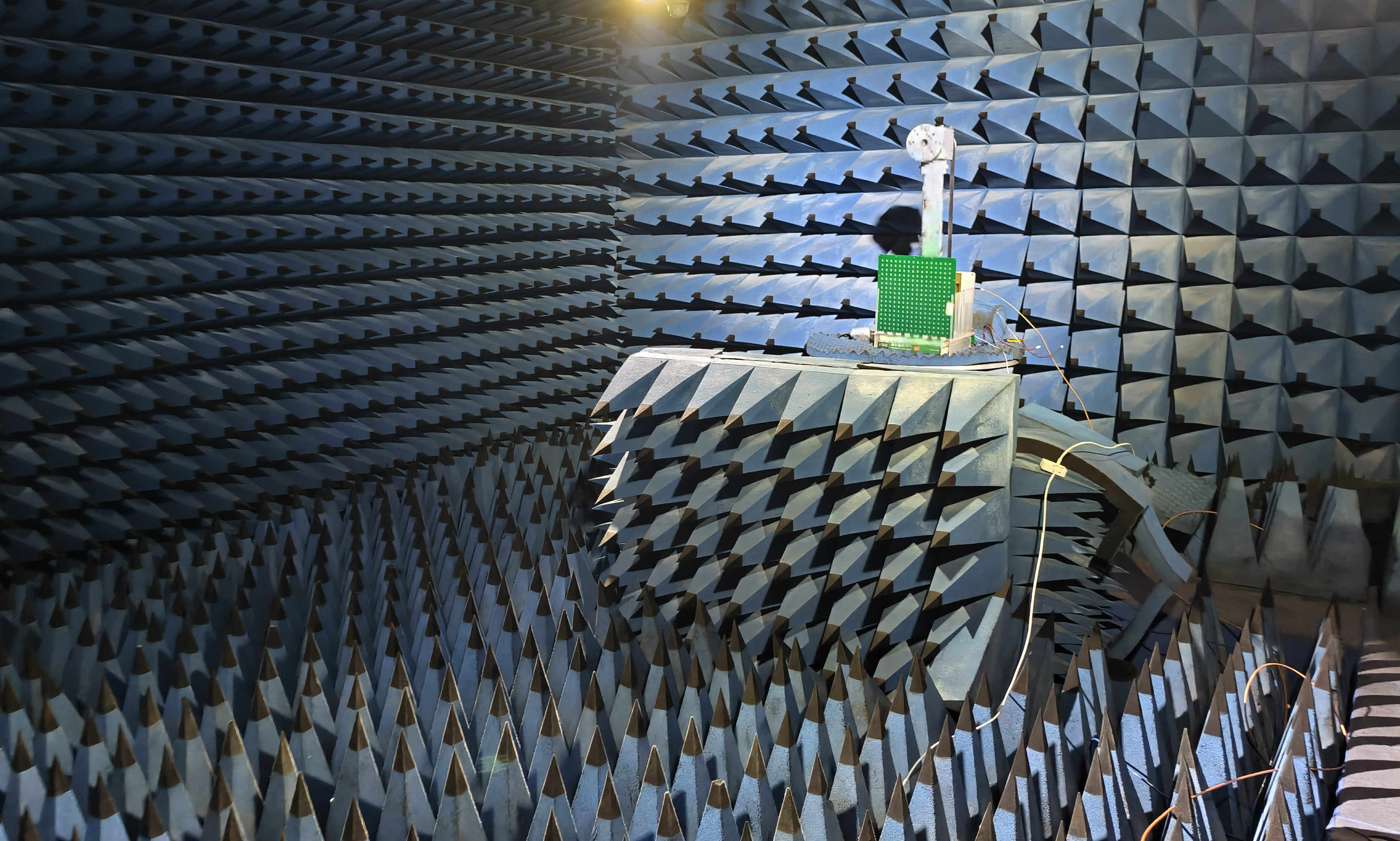}
  \caption{Experiment in the anechoic chamber.}
  \label{Darkroom}
\end{figure}

Two sets of parameters are configured for testing, where the two target communication UEs are supposed at $-20^{\circ}$ and $30^{\circ}$, and the angular suppression intervals are defined as $[0^{\circ},10^{\circ}]$ and $[-10^{\circ},0^{\circ}]$, respectively. The Non-Constraint scheme is adopted as a benchmark, which synthesizes beam patterns solely based on the communication UE directions without incorporating any angular suppression. This comparison evaluates the proposed algorithm's capability to regulate beam radiation within the designated suppression intervals.

\begin{figure}[!htbp]
  \centering
  \subfigure[]{
  \label{0-5}
  \includegraphics[width=0.92\linewidth]{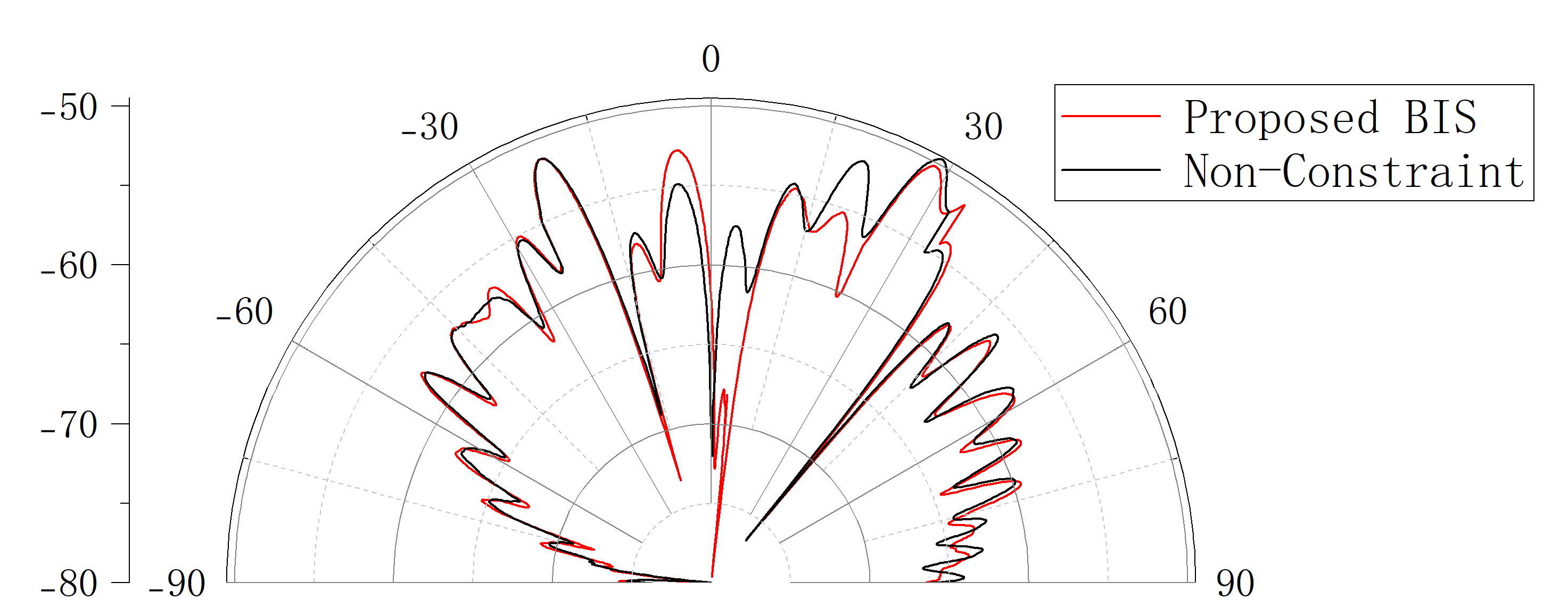}}
  \subfigure[]{
  \label{0-10}
  \includegraphics[width=0.92\linewidth]{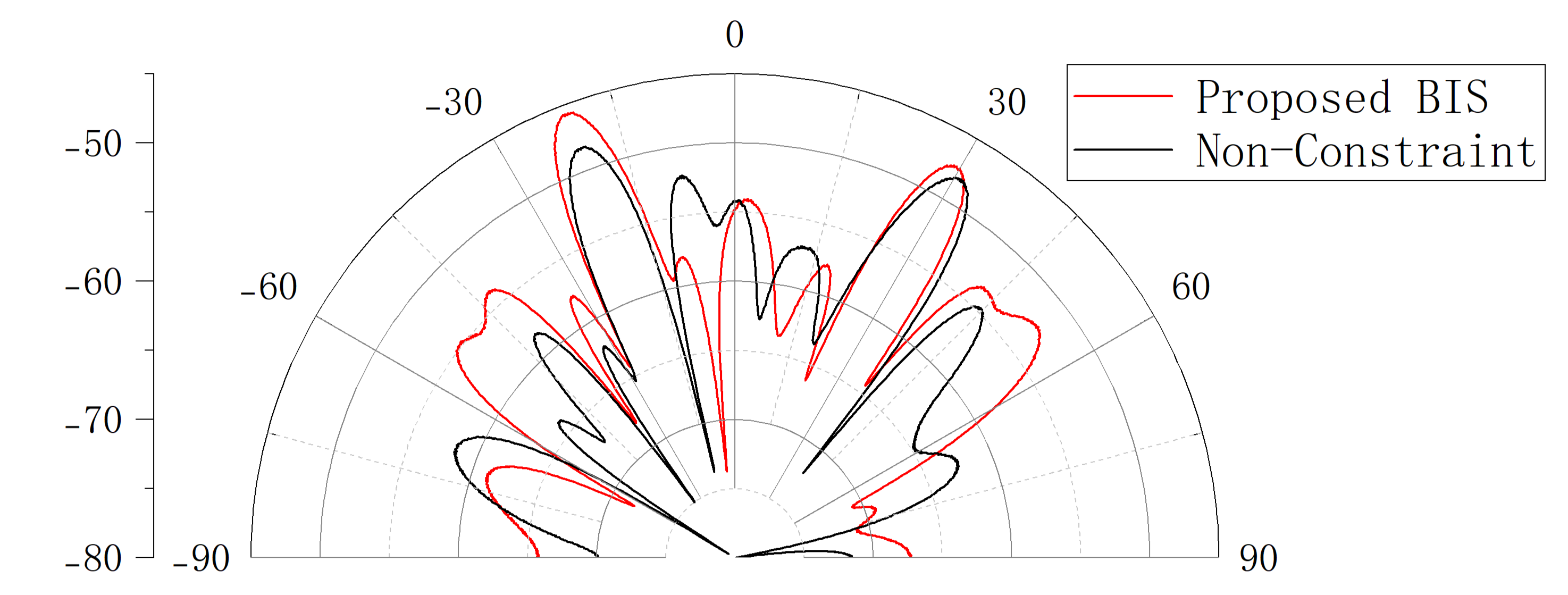}}
  \caption{Measurement in the anechoic chamber. The Tx is fixed behind the transmissive RIS with a distance of 0.6 m, and the recording distance is 5.0 m. The two main beam directions are set to $-20^{\circ}$ and $30^{\circ}$. (a) The suppression interval is set to $[0^{\circ},10^{\circ}]$. (b) The suppression interval is set to $[-10^{\circ},0^{\circ}]$.}
  \label{Measurement}
\end{figure}

The results in Fig.~\ref{Measurement} indicate that in both independent trials, the Non-Constraint scheme effectively generates a highly directional main beam toward the target communication UE directions at $-20^{\circ}$ and $30^{\circ}$. However, the observed beam pattern reveals the presence of significant sidelobes in non-target directions, which could lead to potential information leakage or signal interference.
When the proposed BIS method is applied for directional suppression, a noticeable energy reduction within the suppression intervals is observed. For instance, in Fig.~\ref{0-5}, the power peak at $5^{\circ}$, and in Fig.~\ref{0-10}, the power peak at $-5^{\circ}$, are attenuated by approximately 10 to 20 dB compared to the Non-Constraint.

Note that the beam energy reduction within the designated regions is not completely eliminated. Moreover, while suppressing beam in specific directions, unintended energy enhancements emerge in certain other directions, such as the beams around $-50^{\circ}$ and $50^{\circ}$ in Fig.~\ref{0-10}. These effects can be attributed to two primary factors. First, the RIS used in the experiment is a discrete 1-bit design with only 16 units in the observed plane, resulting in limited phase and beam resolution and an inability to achieve precise beam synthesis and suppression. Second, in accordance with the principle of energy conservation, and given that the transmissive RIS lacks absorption capabilities, suppressing beam energy in certain directions inevitably leads to increased energy in others. These undesirable energy distributions can be mitigated when utilizing high-resolution quantization RISs and considering beam behavior across the entire three-dimensional space.

\section{Conclusion}\label{Section6}
This paper presented a comprehensive framework for flexible beam synthesis and directional suppression using transmissive RIS, through addressing constrained Max-min optimization problems. We developed realistic geometrical optics-based models to characterize the system's input-output behavior and introduced a novel BIS algorithm capable of solving a broad range of constraint Max-min problems. The framework enables flexible beam synthesis, allowing for the enhancement and suppression of beams in specific spatial directions. Simulations and prototype experiments have demonstrated the effectiveness of the proposed method, highlighting its superior performance in beam control and power gains compared to existing algorithms.
The proposed framework exhibits significant potential across a range of applications, including multi-user communication, target detection and tracking, coverage optimization, anti-eavesdropping, interference mitigation, and energy-efficient networking.

\bibliographystyle{IEEEtran}
\bibliography{Reference}
\vfill
\end{document}